\documentclass[draftclsnofoot]{IEEEtran}
\onecolumn
\usepackage{latexsym}
\usepackage{cite}
\usepackage{amssymb}
\usepackage{bm}
\usepackage{amsmath, amssymb}
\usepackage[dvips]{graphics}
\usepackage{graphicx}
\usepackage{psfrag}
 \allowdisplaybreaks

\newtheorem{definition}{Definition}

\newtheorem{theorem}{Theorem}
\newtheorem{lemma}[theorem]{Lemma}

\newtheorem{corollary}[theorem]{Corollary}

\newtheorem{example}{Example}

\begin{document}

\newcommand{\be}{\begin{equation}}
\newcommand{\ee}{\end{equation}}
\newcommand{\bea}{\begin{eqnarray}}
\newcommand{\eea}{\end{eqnarray}}
\newcommand{\beaa}{\begin{eqnarray*}}
\newcommand{\eeaa}{\end{eqnarray*}}

\title{Two-way source coding with a helper}%We'll see

\author{Haim Permuter,  Yossef Steinberg and Tsachy Weissman \\
\thanks{The work of H.
Permuter and T. Weissman is supported by NSF grants 0729119 and
0546535. The work of Y. Steinberg is supported by the ISF (grant No.
280/07).
Author's emails: haimp@bgu.ac.il,  ysteinbe@ee.technion.ac.il, and tsachy@stanford.edu}%
%\thanks{H. Permuter and T. Weissman are with the Department of Electrical Engineering, Stanford University, Stanford, CA 94305, USA.
%(Email: \{haim1, tsachy\}@stanford.edu). J. Chen is with the
%Department of Electrical and Computer Engineering, McMaster
%University, (Email: \{junchen@ece.mcmaster.ca) }
%\thanks{This work was
%supported by the NSF through the CAREER award and TFR-0729119
%grant.}
%\thanks{M. Shell is with the Georgia Institute of Technology.}}
}%
%
%\markboth{Submitted to IEEE Transactions on Information Theory,
%Aug. 2007, revised Aug 2008}{Shell \MakeLowercase{\textit{et
%al.}}: Capacity Region of the Finite-State Multiple Access Channel
%with and without Feedback}

\maketitle \vspace{-1.4cm}
\begin{abstract}%
Consider the two-way rate-distortion problem in which a helper sends
a common limited-rate message to both users based on side
information at its disposal. We characterize the region of
achievable rates and distortions where a Markov form (Helper)-(User
1)-(User 2) holds. The main insight of the result is that in order
to achieve the optimal rate, the helper may use a binning scheme, as
in Wyner-Ziv, where the side information at the decoder is the
``further" user, namely, User 2.   We derive these regions
explicitly for  the Gaussian sources with square error distortion,
analyze a trade-off between the rate from the helper and the rate
from the source, and examine a special case where the helper has the
freedom to send different messages, at different rates, to the
encoder and
the decoder. %We show that ``more help'' to the encoder than to the
%decoder does not yield any performance gain and that in such cases
%the freedom to send different messages to the encoder and the
%decoder yields no gain over the case of a common message. Further,
%in this setting of different messages, the rate to the encoder can
%be strictly less than that to the decoder with no performance
%loss.conclude few
The converse proofs use a new technique for
verifying Markov relations via undirected graphs.
\end{abstract}
\vspace{-0.0cm}
\section{Introduction}
In this paper, we consider the problem of two-way source encoding
with a fidelity criterion in a situation where both users  receive a
common message from a helper.
\begin{figure}[h]{
\psfrag{b1}[][][1]{$X$} \psfrag{box1}[][][1]{User X}
\psfrag{box3}[][][1]{Helper } \psfrag{a2}[][][1]{$R$}
\psfrag{box2}[][][1]{User Z} \psfrag{b3}[][][1]{$\hat X$}
\psfrag{b4}[][][1]{$\hat Z$} \psfrag{Y}[][][1]{$Y$}
\psfrag{Z}[][][1]{$Z$} \psfrag{t1}[][][1]{$R_1$}
\psfrag{t2}[][][1]{$R_2$} \psfrag{t3}[][][1]{$R_3$}
\psfrag{W}[][][1]{$$} \psfrag{X}[][][1]{$$} \psfrag{V}[][][1]{$$}
\psfrag{YZ}[][][1]{$$}\psfrag{D}[][][1]{$$}
%\centerline{\includegraphics[width=13cm]{c:/home/rate_dist_two_way.eps}}
\centerline{\includegraphics[width=12cm]{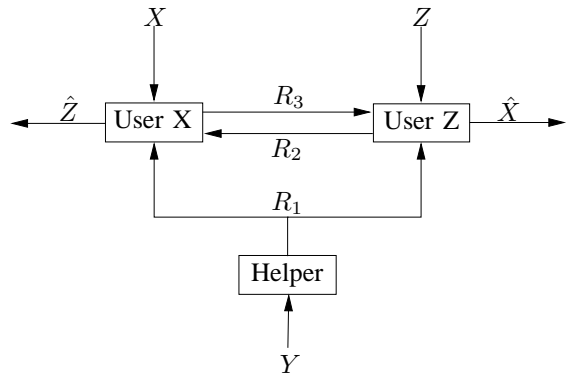}}
\caption{The two-way rate distortion problem with a helper. First
Helper Y sends a common message to User X and to User Z,  then User
Z sends a message to User X, and finally User X sends a message to
User Z. The goal is that User X will reconstruct the sequence $Z^n$
within a fidelity criterion $\mathbb{E}\left[\frac{1}{n}\sum_{i=1}^n
d_z(Z_i,\hat Z_i)\right]\leq D_z$, and User Z will reconstruct the
source $X^n$ within a fidelity criterion
$\frac{1}{n}\mathbb{E}\left[\sum_{i=1}^n d_x(X_i,\hat
X_i)\right]\leq D_x.$ We assume that the side information $Y$ and
the two sources $X,Z$ are i.i.d. and form the Markov chain $Y-X-Z$.}
\label{f_helper_two_way} }\end{figure} The problem is presented in
Fig. \ref{f_helper_two_way}. Note that the case in which the helper
is absent was introduced and solved by Kaspi \cite{Kaspi85_two_way}.

The encoding and decoding is done in  blocks of length $n$. The
communication protocol is that Helper Y first sends a common message
at rate $R_1$ to  User X and to User Z, and then User Z sends a
message at rate $R_2$ to User X, and finally, User X sends a message
to User Z at rate $R_3$. Note that user Z sends his message after it
received only one message, while Sender X sends its message after it
received two messages. We assume that the sources and the helper
sequences are i.i.d. and form the Markov chain $Y-X-Z$. User $X$
receives two messages (one from the helper and one from User Z) and
reconstructs the source $Z^n$. We assume that the fidelity (or
distortion) is of the form $\mathbb{E}\left[\frac{1}{n}\sum_{i=1}^n
d_z(Z_i,\hat Z_i)\right]$ and that this term should be less than a
threshold $D_z$. User $Z$ also receives two messages (one from the
helper and one from User X) and reconstructs the source $X^n$. The
reconstruction $\hat X^n$ must lie within a fidelity criterion of
the form $\frac{1}{n}\mathbb{E}\left[\sum_{i=1}^n d_x(X_i,\hat
X_i)\right]\leq D_x$.

Our main result in this paper is that the achievable region for this
problem is given by $\mathcal R(D_x,D_z)$, which is defined as the
set of all rate triples $(R_1,R_2,R_3)$ that satisfy
\begin{eqnarray}
R_1&\geq& I(Y;U|Z), \label{e_R1_twoway}  \\
R_2&\geq& I(Z;V|U,X),\label{e_R2_twoway} \\
R_3&\geq& I(X;W|U,V,Z),\label{e_R3_twoway}
\end{eqnarray}
for some joint distribution of the form
\begin{equation}
\label{e_p_twoway}
p(x,y,z,u,v,w)=p(x,y)p(z|x)p(u|y)p(v|u,z)p(w|u,v,x),
\end{equation}
where $U$, $V$ and $W$ are auxiliary random variables with bounded
cardinality. The reconstruction variable $\hat Z$ is a deterministic
function of the triple $(U,V,X)$, and the reconstruction  $\hat X$
is a deterministic function of the triple $(U,W,Z)$ such that
\begin{eqnarray}\label{e_def_dist}
\mathbb{E}d_x(X,\hat X(U,V,Z))&\leq& D_x,\nonumber \\
\mathbb{E}d_z(Z,\hat Z(U,W,X))&\leq& D_z.
\end{eqnarray}

%This region can be achieved by a communication scheme that is based
%on the  Wyner-Ziv\cite{WynZiv73} scheme for the case where only the
%decoder has side information.

 The main insight gained from this
region is that the helper may use a code based on binning that is
designed for a decoder with side information, as in Wyner and
Ziv\cite{Wyner_ziv76_side_info_decoder}. User $X$ and User $Z$ do
not have the same side information, but it is sufficient to design
the helper's code assuming that the side information at the decoder
is the one that is ``further" in the Markov chain, namely, $Z$.
Since a distribution of the form (\ref{e_p_twoway}) implies that
$I(U;Z)\leq I(U;X)$, a Wyner-Ziv code at rate $R_1\geq I(Y;U|Z)$
would be decoded successfully both by User Z and by User X. Once the
helper's message has been decoded by both users, a two-way source
coding  is performed where both users have additional side
information $U^n$.

Several papers on related problems have appeared in the past in the
literature. Wyner~\cite{Wyner75_WAK} studied a problem of network
source coding with compressed side information that is provided only
to the decoders. A special case of his model is the system in
Fig.~\ref{f_helper_two_way} but without the memoryless side
information $Z$ and where the stream carrying the  helper's message arrives only
at the decoder (User Z). A full characterization of the achievable
region can be concluded from the results of~\cite{Wyner75_WAK} for
the special case where the source $X$ has to be reconstructed
losslessly. This problem was solved independently by Ahlswede and
K\"{o}rner in~\cite{Ahlswede-Korner75}, but the extension of these
results to the case of lossy reconstruction of $X$ remains open.
Kaspi~\cite{Kaspi79_dissertation} and Kaspi and
Berger~\cite{Kaspi_berger82} derived an achievable region for a
problem that contains the helper problem with degenerate $Z$ as a
special case. However, the converse part does not match.
In~\cite{Vasudevan07_helper}, Vasudevan and Perron described a
general rate distortion problem with encoder breakdown and there
they solved the case where in Fig. \ref{f_helper_two_way} one of the
sources is a constant\footnote{ The case where one of the sources is
constant was also considered independently in
\cite{Permuter_steinber_weissman08_helperunpublished}.}.
%
%Several settings of encoding two correlated sources (a.k.a
%multi-terminal source coding) have been solved in the literature.
%The first case was solved by Slepian and Wolf
%\cite{Slepian_wolf_73_source_coding}, where the goal is to reproduce
%both sources losslessy, and the encoders are ignorant of each
%other's messages. A similar setting was considered by Wyner
%\cite{Wyner75_WAK} and by Ahlswede and K\"{o}rner
%\cite{Ahlswede-Korner75}, also known as WAK, where only one source
%needs to be reconstructed losslessy; the other source is acting as a
%helper. Wyner and Ziv \cite{WynZiv76SideInformationDecoder}
%characterized the rate distortion region of correlated sources when
%one of the rates is unlimited and therefore  the associated source
%is known perfectly to the decoder.
%
%Kaspi and Berger~\cite{Kaspi_berger82} and
%Kaspi~\cite{Kaspi79_dissertation} derived an achievable scheme for a
% general case that contains the regions for all the cases above.
% In particular, case C of Theorem 2.1,
% in \cite{Kaspi_berger82,Kaspi79_dissertation} provides an
% achievable scheme for the problem in Fig. \ref{f_helper} with a general
% distortion of the form $d(x,y,\hat x)$; without a matching converse
%\footnote{In \cite[p.25-26]{Kaspi79_dissertation} there is
% an attempt to prove the converse for the region provided
% in~\cite{Kaspi_berger82}, but the Markov Chain $U-(X,W)-Y$,
% which is one of the conditions characterizing the region, is not proved.}.

Berger and Yeung~\cite{Berger_Yeung89_one_distortion} solved the
multi-terminal source coding problem where one of the two sources
needs to be reconstructed perfectly and the other source needs to be
reconstructed with a fidelity criterion. Oohama
 solved the multi-terminal
source coding case for the two~\cite{Oohama97GaussianMultiterminal}
and $L+1$ \cite{Oohama05_Lhelpers_gaussian} Gaussian sources,  in
which only one source needs to be reconstructed with a mean square
error, that is, the other $L$ sources are helpers. More recently,
Wagner, Tavildar, and  Viswanath characterized the region where both
sources \cite{Wagner08_two_sources} or $L+1$ sources
\cite{Wagner08_L_sources} need to be reconstructed at the decoder
with a mean square error criterion.

In \cite{Kaspi85_two_way}, Kaspi has introduced a multistage
communication between two users, where each user may transmit up to
$K$ messages to the other user that depends on the source and
previous received messages. In this paper we also consider the
multi-stage source coding with a common helper. The case where a
helper is absent and the communication between the users is via
memoryless channels was recently solved by Maor and Merhav
\cite{MaorM_merhav06_two_way} where they showed that a source
channel separation theorem holds.

 The remainder of the paper is organized as follows. In Section \ref{s_preliminary} we present a new technique for verifying Markov
 relations between random variables based on undirected graphs. The technique is used throughout the converse proofs.
  The problem
definition and the achievable region for two way rate distortion
problem with a common helper are presented in Section
\ref{s_definition}. Then we consider two special cases, first in
Section \ref{s_rate_helper_side_information} we consider the case of
$R_2=0$ and $D_z=\infty$, and in Section
\ref{s_wyner_ziv_helper_YZX} we consider $R_3=0$ and $D_x=\infty$.
The proofs of these two special cases provide the insight and the
tricks that are used in the  proof of the  general two-way rate distortion
problem with a helper. The proof of the achievable region for the
two-way rate distortion problem with a helper is given in Section
\ref{s_proof_t_two_way} and it is extended to a multi-stage two way
rate distortion with a helper in Section \ref{s_two_way_multi}. In
Section \ref{s_gaussian_case} we consider the Gauissan instance of
the problem and derive the region explicitly. In Section
\ref{s_wyner_ziv_helper_further} we return to the special case where
$R_2=0$ and $D_z=\infty$ and analyze the trade-off between the bits
from the helper and bits from source and gain insight for the case
where the helper sends different messages to each user, which is an
open problem.

\section{Preliminary: A technique for checking Markov relations \label{s_preliminary}}
Here we present a new technique, based on undirected graphs, that
provides a sufficient condition for establishing a Markov chain from
a joint distribution. We use this technique throughout the paper to
verify Markov relations. A different technique using directed graphs
was introduced by Pearl \cite[Ch 1.2]{Pearl00_causality},
\cite{Kramer03}.

Assume we have a set of random variables $(X_1,X_2,...,X_{N})$,
where $N$ is the size of the set. Without loss of generality, we
assume that the joint distribution has the form
\begin{equation}
p(x^N)=f(x_{\mathcal S_{1}})f(x_{\mathcal S_{2}})\cdots
f(x_{\mathcal S_{K}}),
\end{equation}
where $X_{\mathcal{S}_i} = \{ X_j \}_{j \in \mathcal{S}_i}$, where $\mathcal{S}_i$ is a subset of $\{1,2,\ldots, N\}$. The
following graphical technique provides a sufficient condition for
the Markov relation $X_{\mathcal G_{1}}-X_{\mathcal
G_{2}}-X_{\mathcal G_{3}}$, where $X_{\mathcal G_{i}},\ i=1,2,3$
denote three disjoint subsets of $X^N$.

The technique comprises two steps:
\begin{enumerate}
\item draw an undirected graph where all the random variables $X^N$
are nodes in the graph and for all $i=1,2,..K$ draw edges between
all the nodes $X_{\mathcal S_{i}}$,
\item if all paths in the graph from a node in $X_{\mathcal G_{1}}$ to
a node in $X_{\mathcal G_{3}}$ pass through a node in $X_{\mathcal G_{2}}$,
then the Markov chain $X_{\mathcal G_{1}}-X_{\mathcal
G_{2}}-X_{\mathcal G_{3}}$ holds.
\end{enumerate}
\begin{figure}[h!]{
\psfrag{X1}[][][1]{$X_1$} \psfrag{X2}[][][1]{$X_2$}
\psfrag{Y1}[][][1]{$Y_1$} \psfrag{Y2}[][][1]{$Y_2$}
\psfrag{Z1}[][][1]{$Z_1$} \psfrag{Z2}[][][1]{$Z_2$}
%{C:/cygwin/usr/X11R6/bin/home/rate_helper_side_simpl.eps}}
%\centerline{\includegraphics[width=4cm]{C:/home/markov_chain_example.eps}}
\centerline{\includegraphics[width=4cm]{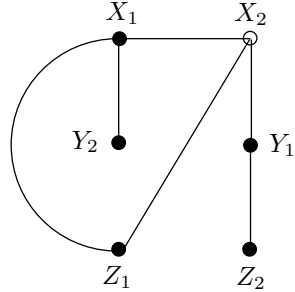}}
\caption{The undirected graph that corresponds to the joint
distribution given in (\ref{e_distr_ex}). The Markov form
$X_1-X_2-Z_2$ holds since all paths from $X_1$ to $Z_2$  pass
through $X_2$. The node with the open circle, i.e., $\circ$, is the
middle term in the Markov chain and all the other nodes are with
solid circles, i.e., $\bullet$.}\label{f_graph_example}}\end{figure}

\begin{example}
Consider the joint distribution
\begin{equation}\label{e_distr_ex}
p(x^2,y^2,z^2)=p(x_1,y_2)p(y_1,x_2)p(z_1|x_1,x_2)p(z_2|y_1).
\end{equation}
Fig. \ref{f_graph_example} illustrates the above technique for
verifying the Markov relation $X_1-X_2-Z_2$. We conclude that since
all the paths from $X_1$ to $Z_2$ pass through $X_2$, the Markov
chain $X_1-X_2-Z_2$ holds.
\end{example}
The proof of the technique is   based on the observation that if
three random variables $X,Y,Z$ have a joint distribution of the form
$p(x,y,z)=f(x,y)f(y,z)$, then the Markov chain $X-Y-Z$ holds. The
proof appears in Appendix \ref{app_proof_technique}.
%Because of space limitation, the proof is omitted.

\section{Problem definitions and main results}
\label{s_definition} Here we formally define the two-way
rate-distortion problem with a helper and present a single-letter
characterization of the achievable region. We use the regular
definitions of rate distortion and we follow the notation of
\cite{CovThom06}. The source sequences $\{X_i\in \mathcal X, \;
i=1,2,\cdots\}$, $\{Z_i\in \mathcal Z, \; i=1,2,\cdots\}$ and the
side information sequence $\{Y_i\in \mathcal Y,\; i=1,2,\cdots\}$
are discrete random variables drawn from finite alphabets $\mathcal
X$, $\mathcal Z$ and $\mathcal Y$, respectively. The random
variables $(X_i,Y_i,Z_i)$ are  i.i.d. $\sim p(x,y,z)$. Let
$\hat{\cal X}$ and $\hat{\cal Z}$ be the reconstruction alphabets,
and $d_x:\ {\cal X}\times{\hat{\cal X}}\rightarrow [0,\infty)$,
$d_z:\ {\cal Z}\times{\hat{\cal Z}}\rightarrow [0,\infty)$ be single
letter distortion measures. Distortion between sequences is defined
in the usual way
\begin{eqnarray}
d(x^n,\hat{x}^n) &=& \frac{1}{n} \sum_{i=1}^n
d(x_i,\hat{x}_i)\nonumber \\
d(z^n,\hat{z}^n) &=& \frac{1}{n} \sum_{i=1}^n d(z_i,\hat{z}_i).
\end{eqnarray}
Let $\mathcal M_i,\ $ denote a set of positive integers
$\{1,2,..,M_i\}$ for $i=1,2,3$.
\begin{definition}\label{def_code}
An $(n,M_1,M_2,M_3,D_x,D_z)$ code for two source $X$ and $Z$ with
helper $Y$ consists of three encoders
\begin{eqnarray}
f_1&:& \mathcal Y^n \to \mathcal M_1 \nonumber \\% \{1,2,...,M_1\} \nonumber \\
f_2 &:& \mathcal Z^n \times \mathcal M_1 \to \mathcal M_2 \nonumber \\
f_3 &:& \mathcal X^n \times \mathcal M_1 \times \mathcal M_2 \to
\mathcal M_3
\end{eqnarray}
%\end{equation}
and two decoders
\begin{eqnarray}
g_2 &:&  \mathcal X^n \times \mathcal M_1 \times \mathcal M_2
\to \hat{\cal Z}^n \nonumber \\
g_3 &:&  \mathcal Z^n  \times \mathcal M_1 \times \mathcal M_3 \to
\hat{\cal X}^n
\end{eqnarray}
such that
\begin{eqnarray}
\mathbb{E}\left[\sum_{i=1}^n d_x(X_i,\hat X_i)\right]&\leq& D_x,\nonumber \\
\mathbb{E}\left[\sum_{i=1}^n d_z(Z_i,\hat Z_i)\right]&\leq& D_z,
\end{eqnarray}
\end{definition}
The rate triple $(R_1,R_2,R_3)$ of the $(n,M_1,M_2,M_3,D_x,D_z)$
code is defined by
\begin{eqnarray}
R_i&=&\frac{1}{n}\log M_i; \;\; i=1,2,3.
\end{eqnarray}

\begin{definition}\label{def_achievable rates}
Given a distortion pair $(D_x,D_z)$, a rate triple $(R_1,R_2,R_3)$
is said to be {\it achievable} if, for any $\epsilon>0$, and
sufficiently large $n$, there exists an
$(n,2^{nR_1},2^{nR_2},2^{nR_3},D_x+\epsilon,D_z+\epsilon)$ code for
the sources $X,Z$ with side information $Y$.
\end{definition}
\begin{definition}\label{def_the achievable_region}
{\it The (operational) achievable region} $\mathcal R^O(D_x,D_z)$ of
rate distortion with a helper known at the encoder and decoder is
the closure of the set of all achievable rate pairs.
\end{definition}
The next theorem is the main result of this work.
\begin{theorem}\label{t_two_way}
In the two way-rate distortion problem with a helper, as depicted in
Fig. \ref{f_helper_two_way}, where $Y-X-Z$,
\begin{equation}
\mathcal R^O(D_x,D_z)=\mathcal R(D_x,D_z),
\end{equation}
where the region $\mathcal R(D_x,D_z)$ is specified in
(\ref{e_R1_twoway})-(\ref{e_def_dist}).
\end{theorem}
Furthermore, the region $\mathcal R(D_x,D_z)$ satisfies the
following properties, which are proved in Appendix
\ref{app_lemma_properties}.

\begin{lemma}\label{lemma:properties_R}
\begin{enumerate}
\item \label{lemma_properties_R_convex}
The region ${\cal R}(D_x,D_z)$ is convex
\item \label{lemma_properties_R_size}
To exhaust ${\cal R}(D_x,D_z)$, it is enough to restrict the
alphabet of $U$, $V$, and $W$ to satisfy
\begin{eqnarray}
|{\cal U}|&\leq& |{\cal Y}|+4,\nonumber \\
|{\cal V}|&\leq& |{\cal Z}||{\cal U}|+3,\nonumber \\
|{\cal W}|&\leq& |{\cal U}||{\cal V}||{\cal X}|+1.
\end{eqnarray}
\end{enumerate}
\end{lemma}

Before proving the main result (Theorem \ref{t_two_way}), we would
like to consider two special cases, first where $R_2=0$ and
$D_z=\infty$ and second where $R_3=0$ and $D_x=\infty$. The main
techniques and insight are gained through those special cases.
Both cases are depicted in Fig. \ref{f_helper_side} where in the
first case we assume the Markov form $Y-X-Z$ and in the second
case we assume a Markov form $Y-Z-X$.

The proofs of these two cases are quite different. In the
achievability of the first case, we use a Wyner-Ziv code that is
designed only for the {\it decoder}, and in the achievability of the
second case we use a Wyner-Ziv code that is designed only for the
{\it encoder}. In the converse for the first case, the main idea is
to observe that the achievable region does not increase by letting
the encoder know $Y$, and in the converse of the second case the
main idea is to use the chain rule in two opposite directions,
conditioning once on the past and once on the future.

\begin{figure}[h!]{
\psfrag{b1}[][][1]{$X$} \psfrag{box1}[][][1]{Encoder}
\psfrag{box3}[][][1]{Helper\ }
 %\psfrag{a2}[][][1]{$T_1(X^n,T_2)$}
 %\psfrag{t1}[][][1]{$\in 2^{nR_1}$}
%\psfrag{b3}[][][1]{$\;\hat X^n(T_1,T_2)$}
%\psfrag{t2}[][][1]{$T_2(Y^n)\in 2^{nR_2}$}
 \psfrag{a2}[][][1]{$R$}
 \psfrag{t1}[][][1]{$$}
\psfrag{A1}[][][1]{$$} \psfrag{A2}[][][1]{$$}
 \psfrag{box2}[][][1]{Decoder}
\psfrag{b3}[][][1]{$\hat X$} \psfrag{a3}[][][1]{}
\psfrag{Y}[][][1]{$Y$} \psfrag{Z}[][][1]{$Z$}
\psfrag{t2}[][][1]{$R_1$}
%\centerline{\includegraphics[width=15cm]
%{C:/cygwin/usr/X11R6/bin/home/rate_distortion_helper_side.eps}}
\centerline{\includegraphics[width=13cm]{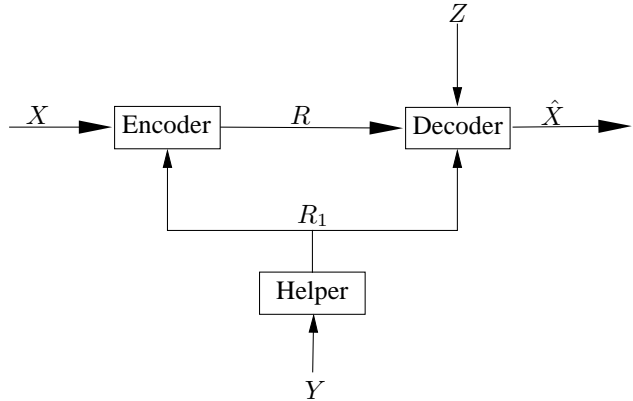}}
\caption{Wyner-Ziv problem with a helper . We consider two cases;
first the source X, Helper Y and the side information Z form  the
Markov chain $Y-X-Z$ and in the second case they form the Markov
chain $Y-Z-X$.} \label{f_helper_side} }\end{figure}

\section{Wyner-Ziv with a helper where Y-X-Z \label{s_rate_helper_side_information}} In this Section,
we consider the rate distortion problem with a helper and additional
side information $Z$, known only to the decoder, as shown in Fig.
\ref{f_helper_side}. We also assume that the source $X$, the helper
$Y$, and the side information $Z$, form the Markov chain $Y-X-Z$.
This setting corresponds to the case where $R_2=0$ and $D_z=\infty$.
Let us denote by ${\cal R}_{Y-X-Z}^O(D)$ the (operational)
achievable region ${\cal R}^O(D_x=D,D_z=\infty)$.% (We are using the
%notation $R$ in Fig. 3 rather than $R_3$, since in the next section
%we consider a different Markov chain $Y-Z-X$ and there $R$ is
%playing the role of $R_2$).

We now present our main result of this section. Let $\mathcal
R_{Y-X-Z}(D)$ be the set of all rate pairs $(R,R_1)$ that satisfy
\begin{eqnarray}
R_1&\geq& I(U;Y|Z),\label{e_R_R_1_ineq_helper_side_2}\\
R&\geq& I(X;W|U,Z), \label{e_R_R_1_ineq_helper_side_1}
\end{eqnarray}
for some joint distribution of the form \bea
\label{e_p_xyuv_decompose}
p(x,y,z,u,v)&=&p(x,y)p(z|x)p(u|y)p(w|x,u),\\
\quad \mathbb{E}d_x(X,\hat X(U,W,Z))&\leq& D,\label{eq:dist2} \eea
where $W$ and $V$ are auxiliary random variables, and the
reconstruction variable $\hat X$ is a deterministic function of the
triple $(U,W,Z)$. The next lemma states properties of $\mathcal
R_{X-Y-Z}(D)$. It is the analog of Lemma~\ref{lemma:properties_R}
and the proof is omitted.
\begin{lemma}\label{lemma:properties_R_SI}
\begin{enumerate}
\item \label{lemma_properties_R_SI_convex}
The region $\mathcal R_{X-Y-Z}(D)$ is convex
\item \label{lemma_properties_R_SI_size}
To exhaust $\mathcal R_{X-Y-Z}(D)$, it is enough to restrict the
alphabets of $V$ and $U$ to satisfy
\begin{eqnarray}
|{\cal U}|&\leq& |{\cal Y}|+2 \nonumber \\
|{\cal W}|&\leq& |{\cal X}|(|{\cal Y}|+2)+1.
\end{eqnarray}
\end{enumerate}
\end{lemma}
%\begin{proof}
%The proof of part~\ref{lemma_properties_R_SI_convex} parallels that
%of part~\ref{lemma_properties_R_convex} of
%Lemma~\ref{lemma:properties_R} and is omitted. For
%part~\ref{lemma_properties_R_SI_size}, using the support
%lemma~\cite{Csiszar81}, the random variable $V$ should have $|{\cal
%Y}|-1$ elements  to preserve $p(y)$, plus three elements to preserve
%$I(V;Y|Z)$, $I(X;U|V,Z)$, and the distortion constraint. Once $V$ is
%fixed, $U$ should have $|{\cal X}||{\cal V}|-1$ elements to preserve
%the joint distribution $p(x,v)$, and two more elements to preserve
%$I(X;U|V,Z)$ and the distortion constraint. This completes the proof
%of the lemma.
%\end{proof}
%
%\begin{eqnarray}\label{e_R_R_1_ineq}
%\end{eqnarray}

\begin{theorem}\label{t_rate_dis_helper_side}
The achievable rate region for the setting illustrated in Fig.
\ref{f_helper_side}, where $X,Y,Z$ are i.i.d. random variables
forming the Markov chain $Y-X-Z$ is
\begin{equation}
{\cal R}_{Y-X-Z}^O(D)=\mathcal R_{Y-X-Z}(D).
\end{equation}
\end{theorem}

%\section{Proof of Theorem \ref{t_rate_dis_side_info}}
Let us define an additional region $\overline{\mathcal
R}_{X-Y-Z}(D)$ the same as ${\mathcal R}_{X-Y-Z}(D)$ but the term
$p(w|x,u)$ in (\ref{e_p_xyuv_decompose}) is replaced by
$p(w|x,u,y)$, i.e.,
\begin{equation} \label{e_p_xyuv_decompose2}
p(x,y,z,u,w)=p(x,y)p(z|x)p(u|y)p(w|x,u,y).
\end{equation}

In the proof of Theorem \ref{t_rate_dis_helper_side}, we show that
$\mathcal R_{Y-X-Z}(D)$ is achievable and that $\overline {\mathcal
R}_{Y-X-Z}(D)$ is an outer bound, and we conclude the proof by
applying the following lemma, which states that the two regions are
equal.
\begin{lemma}\label{l_overline_wyner_ziv}
$\overline{\mathcal R}_{X-Y-Z}(D)=\mathcal R_{X-Y-Z}(D).$
\end{lemma}
\begin{proof}
Trivially we have $\overline{\mathcal R}_{X-Y-Z}(D)\supseteq
{\mathcal R}(D|Z)$. Now we prove that $\overline{\mathcal
R}_{X-Y-Z}(D)\subseteq \mathcal R_{X-Y-Z}(D)$. Let $(R,R_1)\in
\overline{\mathcal R}_{X-Y-Z}(D)$, and
\begin{equation}\label{e_pxyuv_decompose_again}
\overline p(x,y,z,u,w)=p(x,y)p(z|x)p(u|y)\overline p(w|x,u,y)
\end{equation}
be a distribution that satisfies
(\ref{e_R_R_1_ineq_helper_side_2}),(\ref{e_R_R_1_ineq_helper_side_1})
and (\ref{eq:dist2}). Now we show that there exists a distribution
of the form (\ref{e_p_xyuv_decompose}) such that
(\ref{e_R_R_1_ineq_helper_side_1}),(\ref{e_R_R_1_ineq_helper_side_2})
and (\ref{eq:dist2}) hold.

Let
\begin{equation}
p(x,y,z,u,w)=p(x,y,z)p(u|y)\overline p(w|x,u), \label{eq:new_p}
\end{equation}
where $\overline p(w|x,u)$ is induced by $\overline p(x,y,z,u,w)$.
We now show that the terms $I(U;Y|Z)$, $I(X;W|Z,U)$ and
$\mathbb{E}d(X,\hat X(U,W,Z))$ are the same whether we evaluate them
by the joint distribution $p(x,y,z,u,w)$ of~(\ref{eq:new_p}), or by
$\overline p(x,y,z,u,w)$; hence $(R,R_1)\in \mathcal R_{X-Y-Z}(D)$.
In order to show that the terms above are the same it is enough to
show that the marginal distributions $p(y,z,u)$ and $p(x,z,u,w)$
induced by $p(x,y,z,u,w)$ are equal to the marginal distributions
$\overline p(y,z,u)$ and $\overline p(x,z,u,w)$ induced by
$\overline p(x,y,z,u,w)$. Clearly $p(y,u,z)=\overline p(y,u,z)$. In
the rest of the proof we show $p(x,z,u,w)=\overline p(x,z,u,w)$.

A distribution of the form $\overline p(x,y,z,u,w)$ as given in
(\ref{e_pxyuv_decompose_again}) implies that the Markov chain
$W-(X,U)-Z$ holds as shown in Fig. \ref{f_markov3a}.
\begin{figure}[h!]{
\psfrag{X}[][][1]{$X$} \psfrag{Y}[][][1]{$Y$} \psfrag{U}[][][1]{$U$}
\psfrag{Z}[][][1]{$Z$} \psfrag{W}[][][1]{$W$} \psfrag{V}[][][1]{$V$}
%\centerline{\includegraphics[width=6cm]{C:/cygwin/usr/X11R6/bin/home/markov3.eps}}
\centerline{\includegraphics[width=5cm]{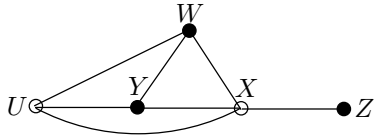}} \caption{A
graphical proof of the Markov chain $W-(X,U)-Z$. The undirected
graph corresponds to the joint distribution given in
(\ref{e_pxyuv_decompose_again}), i.e.,  $\overline
p(x,y,z,u,v,w)=p(x,y)p(z|x)p(u|y)p(w|u,x,y).$ The Markov chain holds
since there is no path from $Z$ to $W$ that does not pass through
$(X,U)$. }\label{f_markov3a}}\end{figure} Therefore $\overline
p(w|x,u,z)=\overline p(w|x,u).$ Now consider $\overline
p(x,z,u,w)=\overline p(x,z,u) \overline p(w|x,u)$, and since
$\overline p(x,z,u)=p(x,z,u)$ and $\overline p(w|x,u)= p(w|x,u)$ we
conclude that $\overline p(x,z,u,w)=p(x,z,u,w)$.
\end{proof}

{\it Proof of Theorem \ref{t_rate_dis_helper_side}:}

{\bf Achievability:} The proof follows classical arguments, and
therefore the technical details will be omitted. We describe only
the coding structure and the associated Markov conditions. Note that
the condition~(\ref{e_p_xyuv_decompose}) in the definition of $
\mathcal R_{X-Y-Z}(D)$, implies the Markov chain $U - Y - X - Z$.
The helper (encoder of $Y$) employs Wyner-Ziv coding with decoder
side information $Z$ and
 external random variable $U$, as seen from~(\ref{e_R_R_1_ineq_helper_side_2}).
The Markov conditions required for such coding, $U - Y - Z$, are
satisfied, hence the source decoder, at the destination, can recover
the codewords constructed from $U$. Moreover,
since~(\ref{e_p_xyuv_decompose}) implies
 $U - Y - X - Z$, the encoder of $X$ can also reconstruct
$U$ (this is the point where the Markov assumption $Y - X - Z$ is
needed). Therefore in the coding/decoding scheme of $X$, $U$ serves
as side information available at both sides. The source ($X$)
encoder now employs Wyner-Ziv coding for $X$, with decoder side
information $Z$, coding random variable $W$, and $U$ available at
both sides. The Markov conditions needed for this scheme are $W -
(X,U) - Z$, which again are satisfied by~(\ref{e_p_xyuv_decompose}).
The rate needed for this coding is $I(X;W|U,Z)$, reflected in the
bound on $R$ in~(\ref{e_R_R_1_ineq_helper_side_1}). Once the two
codes (helper and source code) are decoded, the destination can use
all the available random variables, $U$, $W$, and the side
information $Z$, to construct $\hat{X}$.

{\bf Converse:} Assume that we have an
$(n,M_1=2^{nR_1},M_2=1,M_3=2^{nR},D_x=D,D_z=\infty)$ code as in
Definition \ref{def_code}. We will show the existence of a triple
$(U,W,\hat{X})$ that
satisfy~(\ref{e_R_R_1_ineq_helper_side_2})-(\ref{eq:dist2}). Denote
$T_1=f_1(Y^n)\in\{1,...,2^{nR_1}\}$, and
$T=f_3(X^n,T_1)\in\{1,...,2^{nR}\}$. Then,
\begin{eqnarray}\label{e_conv_HT'_side}
nR_1 &\stackrel{}{\geq}&H(T_1) \nonumber \\
&\stackrel{}{\geq}&H(T_1|Z^n) \nonumber \\
&\stackrel{}{\geq}&I(Y^n;T_1|Z^n) \nonumber \\
&\stackrel{}{=}&\sum_{i=1}^n
H(Y_i|Z_i)-H(Y_i|Y^{i-1},T_1,Z^n)\nonumber \\
&\stackrel{(a)}{=}&\sum_{i=1}^n
H(Y_i|Z_i)-H(Y_i|X^{i-1},Y^{i-1},T_1,Z^n)\nonumber \\
&\stackrel{}{\geq}&\sum_{i=1}^n H(Y_i|Z_i)-H(Y_i|X^{i-1},T_1,Z^n),
\end{eqnarray}
where equality (a) is due to the Markov form
$Y_i-(Y^{i-1},f_1(Y^n),Z^n)-X^{i-1}$. Furthermore,
\begin{eqnarray}\label{e_conv_HT_side}
nR &\stackrel{}{\geq}&H(T) \nonumber \\
&\stackrel{}{\geq}&H(T|T_1,Z^n) \nonumber \\
&\stackrel{}{\geq}&I(X^n;T|T_1,Z^n) \nonumber \\
&\stackrel{}{=}&\sum_{i=1}^n
H(X_i|T_1,Z^n,X^{i-1})-H(X_i|T,T_1,Z^n,X^{i-1})
\end{eqnarray}

Now, let $W_i\triangleq T$ and $U_i\triangleq (X^{i-1},Z^{n\setminus
i},T_1)$, where $Z^{n\setminus i}$ denotes the vector $Z^n$ without
the $i^{th}$ element, i.e., $(Z^{i-1},Z_{i+1}^{n})$. Then
(\ref{e_conv_HT'_side}) and (\ref{e_conv_HT_side}) become
\begin{eqnarray}\label{e_WZ_helper_conv}
R_1&\geq&   \frac{1}{n}\sum_{i=1}^n I(Y_i;U_i|Z_i) \nonumber \\
R&\geq& \frac{1}{n}\sum_{i=1}^n I(X_i; W_i|U_i,Z_i).
\end{eqnarray}

Now we observe that the Markov chain $U_i-Y_i-(X_i,Z_i)$ holds since
we have $(X^{i-1},Z^{n\setminus i},T_1(Y^n))-Y_i-(X_i,Z_i)$. Also
the Markov chain $W_i-(U_i,X_i,Y_i)-Z_i$ holds since
$T(T_1,X^n)-(X^i,Y_i,T_1(Y^n),Z^{n\setminus i})-Z_i$. The
reconstruction at time $i$, i.e.,  $\hat X_i$, is a deterministic
function of $(Z^n,T,T_1)$, and in particular it is a deterministic
function of $(U_i,W_i,Z_i)$. Finally, let $Q$ be a random variable
independent of $X^n,Y^n,Z^n$, and uniformly distributed over the set
$\{1,2,3,..,n\}$. Define the random variables $U\triangleq(Q,U_Q)$,
$W\triangleq (Q,W_Q)$, and  $\hat X \triangleq(\hat X_Q)$ ($\hat X_Q$ is a short notation for time sharing over the
estimators). The
Markov relations $U-Y-(X,Z)$ and $W-(X,U,Y)-Z$, the inequality
$\mathbb E d(X,\hat X)=\sum_{i=1}^n\frac{1}{n} \mathbb E d(X,\hat
X_i)\leq D$, the fact that $\hat X$ is a deterministic function of
$(U,W,Z)$ , and the inequalities $R_1 \geq I(Y;U|Z)$ and $R \geq
I(X,Y;W|U,Z)$ (implied by (\ref{e_WZ_helper_conv})), imply that
$(R,R_1) \in \overline{\mathcal R}_{X-Y-Z}(D)$, completing the proof
by Lemma \ref{l_overline_wyner_ziv}. \hfill \QED

\section{Wyner-Ziv with a helper where $Y-Z-X$\label{s_wyner_ziv_helper_YZX}}
Consider the the rate-distortion problem with side information and
helper as illustrated in Fig. \ref{f_helper_side}, where the random
variables $X,Y,Z$ form the Markov chain $Y-Z-X$. This setting
corresponds to the case where $R_3=0$ and exchanging between $X$ and
$Z$. Let us denote by ${\cal R}_{Y-Z-X}^O(D)$ the (operational)
achievable region.

Let $\mathcal R_{Y-Z-X}(D)$ be the set of all rate pairs $(R,R_1)$
that satisfy
\begin{eqnarray}
R_1&\geq& I(U;Y|X),\label{e_R_R_1_ineq_helper_side_2_yzx} \\
R&\geq& I(X;V|U,Z), \label{e_R_R_1_ineq_helper_side_1_yzx}
\end{eqnarray}
for some joint distribution of the form \bea
\label{e_p_xyuv_decompose_yzx}
p(x,y,z,u,v)&=&p(z,y)p(x|z)p(u|y)p(v|x,u),\\
\quad \mathbb{E}d(X,\hat X(U,V,Z))&\leq& D,\label{eq:dist2_yzx}
\eea where $U$ and $V$ are auxiliary random variables, and the
reconstruction variable $\hat X$ is a deterministic function of
the triple $(U,V,Z)$. The next lemma states properties of ${\cal
R}_{Y-Z-X}(D)$. It is the analog of
Lemma~\ref{lemma:properties_R} and therefore omitted.
\begin{lemma}\label{lemma:properties_R_SI}
\begin{enumerate}
\item \label{lemma_properties_R_SI_convex}
The region ${\cal R}_{Y-Z-X}(D)$ is convex
\item \label{lemma_properties_R_SI_size}
To exhaust ${\cal R}_{Y-Z-X}(D)$, it is enough to restrict the
alphabets of $V$ and $U$ to satisfy
\begin{eqnarray}
|{\cal U}|&\leq& |{\cal Y}|+2 \nonumber \\
|{\cal V}|&\leq& |{\cal X}|(|{\cal Y}|+2)+1.
\end{eqnarray}
\end{enumerate}
\end{lemma}

%\begin{proof}
%The proof of part~\ref{lemma_properties_R_SI_convex} parallels that
%of part~\ref{lemma_properties_R_convex} of
%Lemma~\ref{lemma:properties_R} and is omitted. For
%part~\ref{lemma_properties_R_SI_size}, using the support
%lemma~\cite{Csiszar81}, the random variable $V$ should have $|{\cal
%Y}|-1$ elements  to preserve $p(y)$, plus three elements to preserve
%$I(V;Y|Z)$, $I(X;U|V,Z)$, and the distortion constraint. Once $V$ is
%fixed, $U$ should have $|{\cal X}||{\cal V}|-1$ elements to preserve
%the joint distribution $p(x,v)$, and two more elements to preserve
%$I(X;U|V,Z)$ and the distortion constraint. This completes the proof
%of the lemma.
%\end{proof}

\begin{theorem}\label{t_rate_dis_helper_side_YZX}
The achievable rate region for the setting illustrated in Fig.
\ref{f_helper_side}, where $X_i,Y_i,Z_i$ are i.i.d.  triplets distributed according to the
random  variables $X,Y,Z$ forming the Markov chain $Y-Z-X$ is
\begin{equation}
\mathcal R^O_{Y-Z-X}(D)=\mathcal R_{Y-Z-X}(D).
\end{equation}
\end{theorem}
\begin{proof}

{\bf Achievability:}
 The proof follows classical arguments, and
therefore the technical details will be omitted. We describe only
the coding structure and the associated Markov conditions.  The
helper (encoder of $Y$) employs Wyner-Ziv coding with decoder side
information $X$ and
 external random variable $U$, as seen from~(\ref{e_R_R_1_ineq_helper_side_2_yzx}).
The Markov conditions required for such coding, $U - Y - X$, are
satisfied, hence the source encoder, at the destination, can recover
the codewords constructed from $U$. Moreover,
since~(\ref{e_p_xyuv_decompose_yzx}) implies
 $U - Y - Z - X$, the decoder, at the destination,  can also reconstruct
$U$. Therefore in the coding/decoding scheme of $X$, $U$ serves as
side information available at both sides. The source $X$ encoder now
employs Wyner-Ziv coding for $X$, with decoder side information $Z$,
coding random variable $V$, and $U$ available at both sides. The
Markov conditions needed for this scheme are $V - (X,U) - Z$, which
again are satisfied by~(\ref{e_p_xyuv_decompose_yzx}). The rate
needed for this coding is $I(X;V|U,Z)$, reflected in the bound on
$R$ in~(\ref{e_R_R_1_ineq_helper_side_1_yzx}). Once the two codes
(helper and source code) are decoded, the destination can use all
the available random variables, $U$, $V$, and the side information
$Z$, to construct $\hat{X}$.

{\bf Converse:} Assume that we have a code for a source $X$ with
helper $Y$ and side information $Z$ at rate $(R_1,R)$. We will show
the existence of a triple $(U,V,\hat{X})$ that satisfy~
(\ref{e_R_R_1_ineq_helper_side_2_yzx})-(\ref{eq:dist2_yzx}). Denote
$T_1(Y^n)\in\{1,...,2^{nR_1}\}$, and $T(X^n,T1)\in\{1,...,2^{nR}\}$.
Then,
\begin{eqnarray}\label{e_conv_HT_1side_yzx}
nR_1 &\stackrel{}{\geq}&H(T_1) \nonumber \\
&\stackrel{}{\geq}&H(T_1|X^n) \nonumber \\
&\stackrel{}{\geq}&I(Y^n;T_1|X^n) \nonumber \\
&\stackrel{}{=}&\sum_{i=1}^n
H(Y_i|X_i)-H(Y_i|Y^{i-1},T_1,X^n)\nonumber \\
&\stackrel{(a)}{=}&\sum_{i=1}^n
H(Y_i|X_i)-H(Y_i|Y^{i-1},T_1,X_{i+1}^n,X_i),\nonumber \\
&\stackrel{(b)}{=}&\sum_{i=1}^n
H(Y_i|X_i)-H(Y_i|Y^{i-1},T_1,X_{i+1}^n,X_i,Z^{i-1}),\nonumber \\
&\stackrel{(c)}{\geq}&\sum_{i=1}^n
H(Y_i|X_i)-H(Y_i|T_1,X_{i+1}^n,X_i,Z^{i-1}),
\end{eqnarray}
\begin{figure}[h!]{
\psfrag{X1}[][][1]{$X^{i-1}$} \psfrag{X2}[][][1]{$X_i$}
\psfrag{X3}[][][1]{$X_{i+1}^{n}$} \psfrag{Y1}[][][1]{$Y^{i-1}$}
\psfrag{Y2}[][][1]{$Y_i\;\;\;\;\;$}
\psfrag{Z3}[][][1]{$Z_{i+1}^{n}$} \psfrag{Z1}[][][1]{$Z^{i-1}$}
\psfrag{Z2}[][][1]{$Z_i$} \psfrag{Z3}[][][1]{$Z_{i+1}^n$}
\psfrag{Y3}[][][1]{$Y_{i+1}^{n}$}
\psfrag{T}[][][1]{$T(X^n,T_1)\;\;\;\;\;\;\;\;\;\;\;\;$}
\psfrag{T'}[][][1]{$\;\;\;\;\;\;\;T_1(Y^n)$}
%{C:/cygwin/usr/X11R6/bin/home/rate_helper_side_simpl.eps}}
%\centerline{\includegraphics[width=6cm]{c:/home/markov4.eps}}
\centerline{\includegraphics[width=7cm]{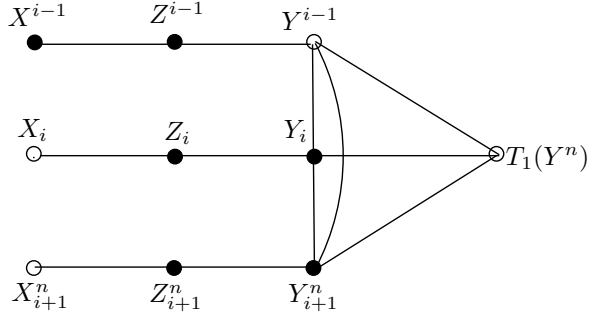}}

\caption{A graphical proof of the Markov chain
$Y_i-(Y^{i-1},T_1(Y^n),X_i^n)-(X^{i-1},Z^{i-1})$. The undirected
graph corresponds to the joint distribution
$p(x^{i-1},z^{i-1})p(y^{i-1}|z^{i-1})p(x_i,z_i)p(y_i|z_i)p(x_{i+1}^n,z_{i+1}^n)p(y_{i+1}^n|z_{i+1}^n)p(t_1|y^n)$.
The Markov chain holds since all paths from $Y_i$ to
$X^{i-1},Z^{i-1}$ pass through $(Y^{i-1},T_1(Y^n),X_i^n)$. The nodes
with the open circle, i.e., $\circ$, constitute the middle term in the
Markov chain, i.e.,  $(Y^{i-1},T_1(Y^n),X_i^n)$ and all the other
nodes are with solid circles, i.e., $\bullet$. The nodes $Y^{i-1}$,
$Y_i$, $Y_{i+1}^n$ and $T_1$ are connected due to the term
$p(t_1|y^n)$.}\label{f_markov4}}\end{figure} where (a) and (b)
follow from the Markov chain
$Y_i-(Y^{i-1},T_1(Y^n),X_i^n)-(X^{i-1},Z^{i-1})$ (see Fig.
\ref{f_markov4} for the proof),
 and (c) follows from the fact that
conditioning reduces entropy. Consider,
\begin{eqnarray}\label{e_conv_side1_yzx}
nR &\stackrel{}{\geq}&H(T) \nonumber \\
&\stackrel{}{\geq}&H(T|T_1,Z^n) \nonumber \\
&\stackrel{}{\geq}&I(X^n;T|T_1,Z^n) \nonumber \\
&\stackrel{}{=}&\sum_{i=1}^n
H(X_i|X_{i+1}^n,T_1,Z^n)-H(X_i|X_{i+1}^n,T_1,Z^n,T)\nonumber \\
&\stackrel{(a)}{=}&\sum_{i=1}^n
H(X_i|X_{i+1}^n,T_1,Z^{i-1},Z_i)-H(X_i|X_{i+1}^n,T_1,Z^n,T)\nonumber \\
&\stackrel{(b)}{\geq}&\sum_{i=1}^n
H(X_i|X_{i+1}^n,T_1,Z^{i-1},Z_i)-H(X_i|X_{i+1}^n,T_1,Z^{i-1},Z_i,T),
\end{eqnarray}
where (a) is due to the Markov chain
$X_i-(X_{i+1}^n,T_1(Y^n),Z^{i})-Z_{i+1}^n$ (this can be seen from
Fig. \ref{f_markov4} since all paths from $X_i$ to  $Z_{i+1}^n$ goes
through $Z_i$), and (b) is due to the fact that conditioning reduces
entropy. Now let us denote $U_i\triangleq
Z^{i-1},T_1(Y^n),X_{i+1}^n$, and $V_i\triangleq T(X^n,T_1)$. The
Markov chains $U_i-Y_i-(X_i,Z_i)$ and $V_i-(X_i,U_i)-(Z_i,Y_i)$ hold
(see Fig. \ref{f_markov5} for the proof of the last Markov
relation).

\begin{figure}[h!]{
\psfrag{X1}[][][1]{$X^{i-1}$} \psfrag{X2}[][][1]{$X_i$}
\psfrag{X3}[][][1]{$X_{i+1}^{n}$} \psfrag{Y1}[][][1]{$Y^{i-1}$}
\psfrag{Y2}[][][1]{$Y_i\;\;\;\;\;$}
\psfrag{Z3}[][][1]{$Z_{i+1}^{n}$} \psfrag{Z1}[][][1]{$Z^{i-1}$}
\psfrag{Z2}[][][1]{$Z_i$} \psfrag{Z3}[][][1]{$Z_{i+1}^n$}
\psfrag{Y3}[][][1]{$Y_{i+1}^{n}$}
\psfrag{T}[][][1]{$T(X^n,T_1)\;\;\;\;\;\;\;\;\;\;\;\;$}
\psfrag{T'}[][][1]{$\;\;\;\;\;\;\;T_1(Y^n)$}
%{C:/cygwin/usr/X11R6/bin/home/rate_helper_side_simpl.eps}}
%\centerline{\includegraphics[width=6cm]{c:/home/markov5.eps}}
\centerline{\includegraphics[width=7cm]{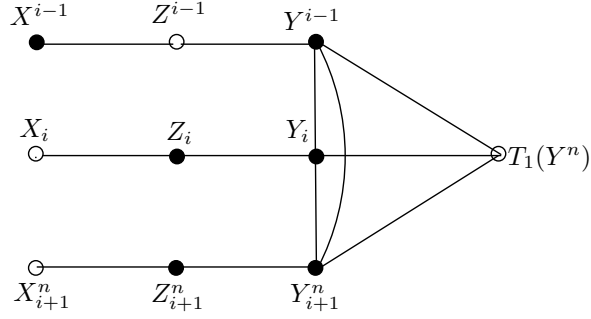}}

\caption{A graphical proof of the Markov chain
$X^{i-1}-(Z^{i-1},T_1(Y^n),X_i^n)-(Z_i,Y_i)$, which implies
$V_i-(X_i,U_i)-(Z_i,Y_i)$. The undirected graph corresponds to the
joint distribution
$p(x^{i-1},z^{i-1})p(y^{i-1}|z^{i-1})p(x_i,z_i)p(y_i|z_i)p(x_{i+1}^n,z_{i+1}^n)p(y_{i+1}^n|z_{i+1}^n)p(t_1|y^n)$.
The Markov chain holds since all paths from $X^{i-1}$ to $(Z_i,Y_i)$
pass through $(Z^{i-1},T_1(Y^n),X_i^n)$.
}\label{f_markov5}}\end{figure} Next, we need to show that there
exists a sequence of function $\hat X_i(U_i,V_i,Z_i)$ such that
\begin{equation}\label{e_cond1}
\frac{1}{n}\sum_{i=1}^n\mathbb{E}[d(X_i,\hat X_i(U_i,V_i,Z_i))]\leq
D.
\end{equation}
 By assumption we know that there exists a sequence of functions $\hat
X_i(T,T_1,Z^n)$ such that $ \sum_{i=1}^n\mathbb{E}[d(X_i,\hat
X_i(T,T_1,Z^n))]\leq nD, $ and trivially this implies that there exists a
sequence of functions $\hat X_i(X^{i-1},T,T_1,Z^n)$ such that
\begin{equation}\label{e_cond_zi+1}
\sum_{i=1}^n\mathbb{E}[d(X_i,\hat
X_i(X_{i+1}^n,T,T_1,Z^i,Z_{i+1}^n))]\leq D.
\end{equation} Note that  the Markov chain
$X_i-(X_{i+1}^n,T_1,Z^{i},T)-Z_{i+1}^n$ holds (see Fig.
\ref{f_markov6} for the proof). Therefore, for an arbitrary function
$\tilde{f}$ of the form $\tilde{f}(X_{i+1}^n,T_1,Z^i,T)$ we have
\begin{equation} \label{e_existance_min}\sum_{i=1}^n\mathbb{E}[d(X_i,\hat
X_i(X_{i+1}^n,T,T_1,Z^i,Z_{i+1}^n))] \leq \min_{\tilde{f}}
\sum_{i=1}^n\mathbb{E}[d(X_i,\hat
X_i(X_{i+1}^n,T,T_1,Z^i,\tilde{f}(X_{i+1}^n,T_1,Z^i,T)))],
\end{equation}
and since each summand on the RHS of (\ref{e_existance_min}) includes only the
random variables $(X_{i+1}^n,T,T_1,Z^i)$ we conclude that there
exists a sequence of functions $\{ X_i(X_{i+1}^n,T,T_1,Z^i)\}$ for
which (\ref{e_cond1}) holds.

\begin{figure}[h!]{
\psfrag{X1}[][][1]{$X^{i-1}$} \psfrag{X2}[][][1]{$X_i$}
\psfrag{X3}[][][1]{$X_{i+1}^{n}$} \psfrag{Y1}[][][1]{$Y^{i-1}$}
\psfrag{Y2}[][][1]{$Y_i$} \psfrag{Z3}[][][1]{$Z_{i+1}^{n}$}
\psfrag{Z1}[][][1]{$Z^{i-1}$} \psfrag{Z2}[][][1]{$Z_i$}
\psfrag{Z3}[][][1]{$Z_{i+1}^n$} \psfrag{Y3}[][][1]{$Y_{i+1}^{n}$}
\psfrag{T}[][][1]{$T(X^n,T_1)\;\;\;\;\;\;\;\;\;\;\;\;$}
\psfrag{T'}[][][1]{$\;\;\;\;\;\;\;T_1(Y^n)$}
%{C:/cygwin/usr/X11R6/bin/home/rate_helper_side_simpl.eps}}
%\centerline{\includegraphics[width=7cm]{markov6.eps}}
%\centerline{\includegraphics[width=7cm]{c:/home/markov6.eps}}
\centerline{\includegraphics[width=7cm]{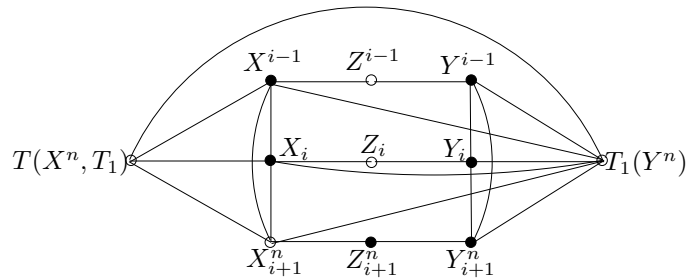}}

\caption{A graphical proof of the Markov chain
$X_i-(X_{i+1}^n,T_1,Z^{i},T)-Z_{i+1}^n$. The undirected graph
corresponds to the joint distribution
$p(x^{i-1},z^{i-1})p(y^{i-1}|z^{i-1})p(x_i,z_i)p(y_i|z_i)p(x_{i+1}^n,z_{i+1}^n)p(y_{i+1}^n|z_{i+1}^n)p(t_1|y^n)p(t|x^n,t_1)$.
The Markov chain holds since all paths from $X_i$ to $Z_{i+1}^n$
pass through $(X_{i+1}^n,T_1,Z^{i},T)$.
}\label{f_markov6}}\end{figure}  Finally, let $Q$ be a random
variable independent of $X^n,Y^n,Z^n$, and uniformly distributed
over the set $\{1,2,3,..,n\}$. Define the random variables
$U\triangleq(Q,U_Q)$, $W\triangleq (Q,W_Q)$, and  $\hat X \triangleq
\hat X_Q$ ($\hat X_Q$ is a short notation for time sharing over the
estimators). Then (\ref{e_conv_HT_1side_yzx})-(\ref{e_cond1})
implies that
(\ref{e_R_R_1_ineq_helper_side_2_yzx})-(\ref{eq:dist2_yzx}) hold.
\end{proof}

\section{Proof of Theorem \ref{t_two_way}\label{s_proof_t_two_way}}

In this section we prove Theorem \ref{t_two_way}, which states that
the (operational) achievable region $\mathcal R^O(D_x,D_z)$ of the
two-way source coding with helper problem as in Fig.
\ref{f_helper_two_way} equals $\mathcal R(D_x,D_z)$.  In the
converse proof we use the ideas used in proving the converses of
Theorems \ref{t_rate_dis_helper_side} and
\ref{t_rate_dis_helper_side_YZX}. Namely, we will use the chain rule
based on the past and future, and will show that $\mathcal
R^O(D_x,D_z)\subseteq \overline{\mathcal R}(D_x,D_z)$, where
$\overline{\mathcal R}(D_x,D_z)$ is  defined as ${\mathcal
R}(D_x,D_z)$ in (\ref{e_R1_twoway})-(\ref{e_def_dist}) but with one
difference: the term $p(w|u,v,x)$ in (\ref{e_p_twoway}) should be
replaced by $p(w|u,v,x,y)$, i.e.,
\begin{equation} \label{e_p_twoway_upper}
p(x,y,z,u,v,w)=p(x,y)p(z|x)p(u|y)p(v|u,z)p(w|u,v,x,y).
\end{equation} The following lemma states that the two regions
$\overline{\mathcal R}(D_x,D_z)$ and ${\mathcal R}(D_x,D_z)$ are
equal.
\begin{lemma}
$\overline{\mathcal R}(D_x,D_z)={\mathcal R}(D_x,D_z).$
\end{lemma}
\begin{proof}
Trivially we have $\overline{\mathcal R}(D_x,D_z)\supseteq {\mathcal
R}(D_z,D_z)$. Now we prove that $\overline{\mathcal
R}(D_x,D_z)\subseteq {\mathcal R}(D_x,D_z)$. Let $(R_1,R_2,R_3)\in
\overline{\mathcal R}(D_x,D_z)$, and
\begin{equation}\label{e_two_way_again}
\overline p(x,y,z,u,v,w)=p(x,y)p(z|x)p(u|y)p(v|u,z)\overline
p(w|u,v,x,y),
\end{equation}
be a distribution that satisfies
(\ref{e_R1_twoway})-(\ref{e_R3_twoway}) and (\ref{e_def_dist}). Next
we show that there exists a distribution of the form of
(\ref{e_p_twoway}) (which is explicitly given in
(\ref{e_two_way_again})) such that
(\ref{e_R1_twoway})-(\ref{e_R3_twoway}) and (\ref{e_def_dist}) hold.
Let
\begin{equation}\label{e_two_way_again2}
p(x,y,z,u,v,w)=p(x,y)p(z|x)p(u|y)p(v|u,z)\overline p(w|u,v,x),
\end{equation}
where $\overline p(w|u,v,x)$ is induced by $\overline p(x,y,z,u,v)$.
We show that all the terms in
(\ref{e_R1_twoway})-(\ref{e_R3_twoway}) and (\ref{e_def_dist}) i.e.,
$I(Y;U|Z)$, $I(Z;V|U,X)$, $\mathbb{E}d_z(Z,\hat Z(U,V,X))$,
$I(X;W|U,V,Z)$, and $\mathbb{E}d_x(X,\hat X(U,W,Z))$ are the same
whether we evaluate them by the joint distribution $p(x,y,z,u,v)$
of~(\ref{e_two_way_again2}), or by $\overline p(x,y,z,u,v,w)$ of
(\ref{e_two_way_again}); hence $(R_1,R_2,R_3)\in {\mathcal
R}(D_x,D_z)$. In order to show that the terms above are the same it
is enough to show that the marginal distributions $p(x,y,z,u,v)$ and
$p(x,z,u,v,w)$ induced by $p(x,y,z,u,v,w)$ are equal to the marginal
distributions $\overline p(x,y,z,u,v)$ and $\overline p(x,z,u,v,w)$
induced by $\overline p(x,y,z,u,v,w)$. Clearly
$p(x,y,z,u,v)=\overline p(x,y,z,u,v)$. In the rest of the proof we
show $p(x,z,u,v,w)=\overline p(x,z,u,v,w)$.

\begin{figure}[h!]{
\psfrag{X}[][][1]{$X$} \psfrag{Y}[][][1]{$Y$} \psfrag{U}[][][1]{$U$}
\psfrag{Z}[][][1]{$Z$} \psfrag{W}[][][1]{$W$} \psfrag{V}[][][1]{$V$}
%\centerline{\includegraphics[width=6cm]{C:/cygwin/usr/X11R6/bin/home/markov3.eps}}
\centerline{\includegraphics[width=5cm]{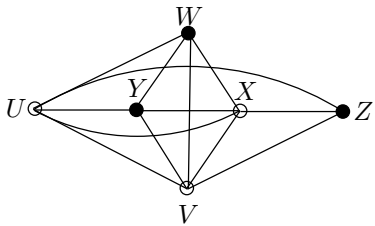}} \caption{A
graphical proof of the Markov chain $W-(X,U,V)-Z$. The undirected
graph corresponds to the joint distribution given in
(\ref{e_two_way_again}), i.e.,  $\overline
p(x,y,z,u,v,w)=p(x,y)p(z|x)p(u|y)p(v|u,z)\overline p(w|u,v,x,y).$
The Markov chain holds since there is no path from $Z$ to $W$ that
does not pass through $(X,U,V)$. }\label{f_markov3}}\end{figure}

A distribution of the form $\overline p(x,y,z,u,v,w)$ as given in
(\ref{e_two_way_again}) implies that the Markov chain $W-(X,U,V)-Z$
holds (see Fig. \ref{f_markov3} for the proof). Therefore $\overline
p(w|u,x,v,z)=\overline p(w|u,x,v).$ Since $\overline
p(x,z,u,v,w)=\overline p(x,z,v,u) \overline p(w|x,u,v)$, and since
$\overline p(x,z,v,u)=p(x,z,v,u)$ and $\overline p(w|x,u,v)=
p(w|x,w,v)$ we conclude that $\overline p(x,z,u,v,w)=p(x,z,u,v,w)$.
\end{proof}

%\begin{theorem}\label{t_rate_dis_helper_side}
%The achievable rate region for the setting illustrated in Fig.
%\ref{f_helper_side}, where $X,Y,Z$ are i.i.d random variables
%forming the Markov chain $Y-Z-X$ is
%\begin{equation}
%\mathcal R^O_{Y-Z-X}(D)=\mathcal R_{Y-Z-X}(D).
%\end{equation}
%\end{theorem}
{\it proof of Theorem \ref{t_two_way}:}

{\bf Achievability:} The achievability scheme is based on the fact
that for the two special cases considered above, namely $R_2=0$ and
$R_3=0$, the coding scheme for the helper was based on a Wyner-Ziv
scheme, where the side information at the decoder is the random
variable that is "further" in the Markov chain $Y-X-Z$, namely $Z$.
The helper (encoder of $Y$) employs Wyner-Ziv coding with decoder
side information $Z$ and
 external random variable $U$, as seen from~(\ref{e_R1_twoway}), i.e., $R_1\geq I(Y;U|Z)$.
The Markov conditions required for such coding, $U - Y - Z$, are
satisfied, hence the source decoder, at the destination, can recover
the codewords constructed from $U$. Moreover,
since~(\ref{e_p_xyuv_decompose_yzx}) implies
 $U - Y - Z - X$, the encoder of $X$ can also reconstruct
$U$. Therefore in the coding/decoding scheme of $X$, $U$ serves as
side information available at both sides. The source $Z$ encoder now
employs Wyner-Ziv coding for $Z$, with decoder side information $X$,
coding random variable $V$, and $U$ available at both sides. The
Markov conditions needed for this scheme are $V - (X,U) - Z$, which
again are satisfied by~(\ref{e_p_twoway}). The rate needed for this
coding is $I(X;V|U,Z)$, reflected in the bound on $R_2$
in~(\ref{e_R2_twoway}). Finally, the source $X$ encoder now employs
Wyner-Ziv coding for $X$, with decoder side information $Z$, coding
random variable $W$, and $U,V$ available at both sides. The Markov
conditions needed for this scheme are $W - (X,U,V) - Z$, which again
are satisfied by~(\ref{e_p_twoway}). The rate needed for this coding
is $I(X;W|U,V,Z)$, reflected in the bound on $R_3$
in~(\ref{e_R3_twoway}). Once the  codes are decoded, the destination
can use all the available random variables, ($U,V,X$) at User X,
and, ($U,W,Z$) at User Z, to construct $\hat{Z}$ and $\hat{X}$,
respectively.

{\bf Converse:} Assume that we have a $(n,M_1,M_2,M_3,D_x,D_z)$
code. We now show the existence of a triple
$(U,V,W,\hat{X},\hat{Z})$ that
satisfy~(\ref{e_R1_twoway})-(\ref{e_def_dist}). Denote
$T_1=f_1(Y^n)$, $T_2=f_2(Z^n,T_1)$, and $T_3=f_3(X^n,T_2,T_1)$. Then using
the same arguments as  in (\ref{e_conv_HT_1side_yzx}) and
(\ref{e_conv_side1_yzx}) (just exchanging between $X$ and $Z$), we
obtain
\begin{eqnarray}\label{e_R1_tow_way_con}
nR_1 &\stackrel{}{\geq}&\sum_{i=1}^n
H(Y_i|Z_i)-H(Y_i|X^{i-1},T_1,Z_i^n),
\end{eqnarray}%
%where step (a) follows the same arguments as  in
%(\ref{e_conv_HT_1_side}) and step (b) follows from the fact that
%conditioning reduces entropy. Using similar arguments as in
%(\ref{e_conv_side1}) (just exchanging between $X$ and $Z$) we
%obtain
\begin{eqnarray}\label{e_R2_tow_way_con}
nR_2 &\stackrel{}{\geq} & \sum_{i=1}^n
H(Z_i|Z_{i+1}^n,T_1,X^{i-1},X_i)-H(Z_i|Z_{i+1}^n,T_1,X^{i-1},X_i,T_2),
\end{eqnarray}
respectively. For upper-bounding $R_3$, consider
\begin{eqnarray}\label{e_R3_tow_way_con}
nR_3 &\stackrel{}{\geq}&H(T_3)\nonumber \\
&\stackrel{}{\geq}&H(T_3|T_1,T_2,Z^n)\nonumber \\
&\stackrel{}{\geq}&I(X^n;T_3|T_1,T_2,Z^n)\nonumber \\
&\stackrel{}{=}&\sum_{i=1}^n H(X_i|X^{i-1},Z^n,T_1,T_2)-H(X_i|X^{i-1},Z^n,T_1,T_2,T_3) \nonumber \\
&\stackrel{(a)}{=}&\sum_{i=1}^n H(X_i|X^{i-1},Z_i^n,T_1,T_2)-H(X_i|X^{i-1},Z^n,T_1,T_2,T_3) \nonumber \\
&\stackrel{}{\geq}&\sum_{i=1}^n
H(X_i|X^{i-1},Z_i^n,T_1,T_2)-H(X_i|X^{i-1},Z_i^n,T_1,T_2,T_3),
\end{eqnarray}
where equality (a) is due to the Markov chain
$X_i-(X^{i-1},Z_i^n,T_1,T_2)-Z^{i-1}$ (see Fig. \ref{f_markov7}).
\begin{figure}[h!]{
\psfrag{X1}[][][1]{$X^{i-1}$} \psfrag{X2}[][][1]{$X_i$}
\psfrag{X3}[][][1]{$X_{i+1}^{n}$} \psfrag{Y1}[][][1]{$Y^{i-1}$}
\psfrag{Y2}[][][1]{$Y_i$} \psfrag{Z3}[][][1]{$Z_{i+1}^{n}$}
\psfrag{Z1}[][][1]{$Z^{i-1}$} \psfrag{Z2}[][][1]{$Z_i$}
\psfrag{Z3}[][][1]{$Z_{i+1}^n$} \psfrag{Y3}[][][1]{$Y_{i+1}^{n}$}
\psfrag{T'}[][][1]{$\;\;\;\;\;\;\;\;\;T_1(Y^n)$}
\psfrag{T}[][][1]{$T_2(Z^n,T_1)\;\;\;\;\;\;\;\;\;\;\;\;\;\;$}
%{C:/cygwin/usr/X11R6/bin/home/rate_helper_side_simpl.eps}}
%\centerline{\includegraphics[width=7cm]{markov6.eps}}
%\centerline{\includegraphics[width=7cm]{c:/home/markov7.eps}}
\centerline{\includegraphics[width=7cm]{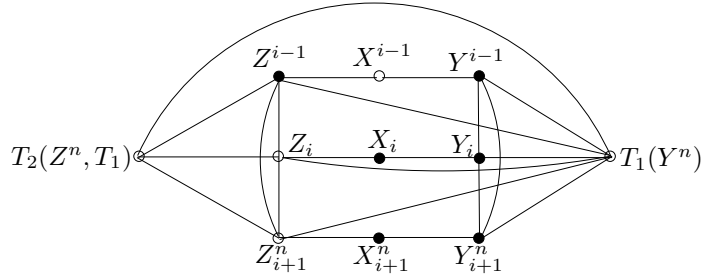}}

\caption{A graphical proof of the Markov chain
$X_i-(X^{i-1},Z_i^n,T_1,T_2)-Z^{i-1}$. The undirected graph
corresponds to the joint distribution
$p(x^{i-1},z^{i-1})p(y^{i-1}|x^{i-1})p(x_i,z_i)p(y_i|x_i)p(x_{i+1}^n,z_{i+1}^n)p(y_{i+1}^n|x_{i+1}^n)p(t_1|y^n)p(t_2|z^n,t_1)$.
The Markov chain holds since all paths from  $Z^{i-1}$ to $X_i$ pass
through $(X^{i-1},Z_i^n,T_1,T_2)$. }\label{f_markov7}}\end{figure}
Now let us denote $U_i\triangleq X^{i-1},T_1,Z_{i+1}^n$,
$V_i\triangleq T_2$ and $W_i\triangleq T_3$, and we obtain from
(\ref{e_R1_tow_way_con})-(\ref{e_R3_tow_way_con})
\begin{eqnarray}\label{e_conv_two_way_i}
R_1&\geq& \frac{1}{n} \sum_{i=1}^n I(Y_i;U_i|Z_i), \nonumber   \\
R_2&\geq& \frac{1}{n} \sum_{i=1}^n I(Z_i;V_i|U_i,X_i),\nonumber  \\
R_3&\geq& \frac{1}{n} \sum_{i=1}^n I(X_i;W_i|U_i,V_i,Z_i),
\end{eqnarray}
Now, we verify that the joint distribution of
$(X_i,Y_i,Z_i,U_i,V_i,W_i)$ is of the form (\ref{e_p_twoway_upper}),
i.e., $U_i-Y_i-(Z_i,X_i)$, $V_i-(U_i,Z_i)-(Y_i,X_i)$ and
$W_i-(U_i,V_i,X_i,Y_i)-Z_i$, hold. The Markov chain
$(T_1(Y^n),X^{i-1},Z_{i+1}^n)-Y_i-(Z_i,X_i)$ trivially holds, and
the Markov chains
\begin{equation}
Z^{i-1}-(T_1(Y^n),X^{i-1},Z_{i}^n)-(Y_i,X_i),
\end{equation}
\begin{equation}X_{i+1}^n-(T_1(Y^n),T_2(T_1,Z^n),X^i,Z_{i+1}^n,Y_i)-Z_i
\end{equation} are proven in is proven in Fig. \ref{f_markov8},
\ref{f_markov9}, respectively.
\begin{figure}[h!]{
\psfrag{X1}[][][1]{$X^{i-1}$} \psfrag{X2}[][][1]{$X_i$}
\psfrag{X3}[][][1]{$X_{i+1}^{n}$} \psfrag{Y1}[][][1]{$Y^{i-1}$}
\psfrag{Y2}[][][1]{$Y_i$} \psfrag{Z3}[][][1]{$Z_{i+1}^{n}$}
\psfrag{Z1}[][][1]{$Z^{i-1}$} \psfrag{Z2}[][][1]{$Z_i$}
\psfrag{Z3}[][][1]{$Z_{i+1}^n$} \psfrag{Y3}[][][1]{$Y_{i+1}^{n}$}
\psfrag{T'}[][][1]{$\;\;\;\;\;\;\;\;\;T_1(Y^n)$}
\psfrag{T}[][][1]{$T_2(Z^n,T_1)\;\;\;\;\;\;\;\;\;\;\;\;\;\;$}
%\centerline{\includegraphics[width=5cm]{c:/home/markov8.eps}}
\centerline{\includegraphics[width=7cm]{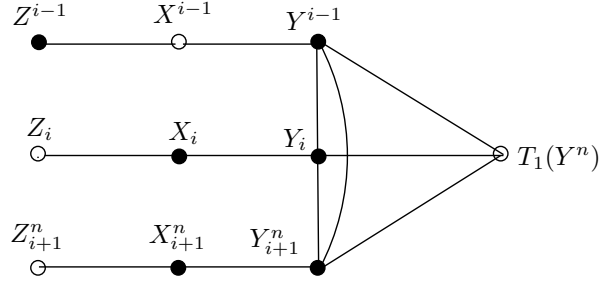}}

\caption{A graphical proof of the Markov chain
$Z^{i-1}-(T_1(Y^n),X^{i-1},Z_{i}^n)-(Y_i,X_i)$. The undirected graph
corresponds to the joint distribution
$p(x^{i-1},z^{i-1})p(y^{i-1}|x^{i-1})p(x_i,z_i)p(y_i|x_i)p(x_{i+1}^n,z_{i+1}^n)p(y_{i+1}^n|x_{i+1}^n)p(t_1|y^n)$.
The Markov chain holds since all paths from  $Z^{i-1}$ to
$(X_i,Y_i)$ pass through $(X^{i-1},Z_i^n,T_1)$.
}\label{f_markov8}}\end{figure}
\begin{figure}[h!]{
\psfrag{X1}[][][1]{$X^{i-1}$} \psfrag{X2}[][][1]{$X_i$}
\psfrag{X3}[][][1]{$X_{i+1}^{n}$} \psfrag{Y1}[][][1]{$Y^{i-1}$}
\psfrag{Y2}[][][1]{$Y_i$} \psfrag{Z3}[][][1]{$Z_{i+1}^{n}$}
\psfrag{Z1}[][][1]{$Z^{i-1}$} \psfrag{Z2}[][][1]{$Z_i$}
\psfrag{Z3}[][][1]{$Z_{i+1}^n$} \psfrag{Y3}[][][1]{$Y_{i+1}^{n}$}
\psfrag{T'}[][][1]{$\;\;\;\;\;\;\;\;\;T_1(Y^n)$}
\psfrag{T}[][][1]{$T_2(Z^n,T_1)\;\;\;\;\;\;\;\;\;\;\;\;\;\;$}
%\centerline{\includegraphics[width=7cm]{c:/home/markov9.eps}}
\centerline{\includegraphics[width=7cm]{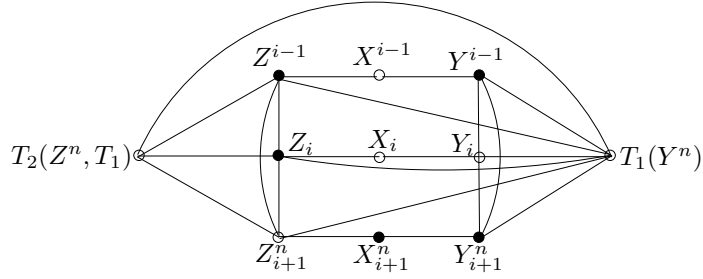}}

\caption{A graphical proof of the Markov chain
$X_{i+1}^n-(T_1(Y^n),T_2(T_1,Z^n),X^i,Z_{i+1}^n,Y_i)-Z_i$. The
undirected graph corresponds to the joint distribution
$p(x^{i-1},y^{i-1})p(z^{i-1}|y^{i-1})p(x_i,y_i)p(z_i|y_i)p(x_{i+1}^n,y_{i+1}^n)p(z_{i+1}^n|y_{i+1}^n)p(t_1|y^n)p(t_2|z^n,t_1)$.
The Markov chain holds since all paths from  $Z^{i}$ to $X_{i+1}^n$
pass through $(T_1(Y^n),T_2(T_1,Z^n),X^i,Z_{i+1}^n,Y_i)$.
}\label{f_markov9}}\end{figure} Next, we show that exist sequences
of functions $\{\hat Z_i(U_i,W_i,Z_i)\}$, and $\{\hat
X_i(U_i,V_i,Z_i)\}$ such that
\begin{eqnarray}\label{e_cond1_two_way}
\frac{1}{n}\sum_{i=1}^n\mathbb{E}[d(X_i,\hat
X_i(U_i,V_i,Z_i))]&\leq& D_x, \nonumber \\
\frac{1}{n}\sum_{i=1}^n\mathbb{E}[d(X_i,\hat
Z_i(U_i,W_i,X_i))]&\leq& D_z.
\end{eqnarray}
The only difficulty here is that the terms in $(U_i,V_i,Z_i)$ do not
include $Z^{i-1}$ and the terms $(U_i,W_i,X_i)$ do not include
$X_{i+1}^n$. However, this is solved by the same argument as for the
Wyner-Ziv with helper at the end of  Section
\ref{s_wyner_ziv_helper_YZX}, by showing the Markov forms
$X_i-(U_i,V_i,Z_i)-Z^{i-1}$  and $Z_i-(U_i,W_i,X_i)-X_{i+1}^n$ for
which the proofs are given in Figures \ref{f_markov10} and
\ref{f_markov11}, respectively.

\begin{figure}[h!]{
\psfrag{X1}[][][1]{$X^{i-1}$} \psfrag{X2}[][][1]{$X_i$}
\psfrag{X3}[][][1]{$X_{i+1}^{n}$} \psfrag{Y1}[][][1]{$Y^{i-1}$}
\psfrag{Y2}[][][1]{$Y_i$} \psfrag{Z3}[][][1]{$Z_{i+1}^{n}$}
\psfrag{Z1}[][][1]{$Z^{i-1}$} \psfrag{Z2}[][][1]{$Z_i$}
\psfrag{Z3}[][][1]{$Z_{i+1}^n$} \psfrag{Y3}[][][1]{$Y_{i+1}^{n}$}
\psfrag{T'}[][][1]{$\;\;\;\;\;\;\;\;\;T_1(Y^n)$}
\psfrag{T}[][][1]{$T_2(Z^n,T_1)\;\;\;\;\;\;\;\;\;\;\;\;\;\;$}
%{C:/cygwin/usr/X11R6/bin/home/rate_helper_side_simpl.eps}}
%\centerline{\includegraphics[width=7cm]{markov6.eps}}
%\centerline{\includegraphics[width=7cm]{c:/home/markov10.eps}}
\centerline{\includegraphics[width=7cm]{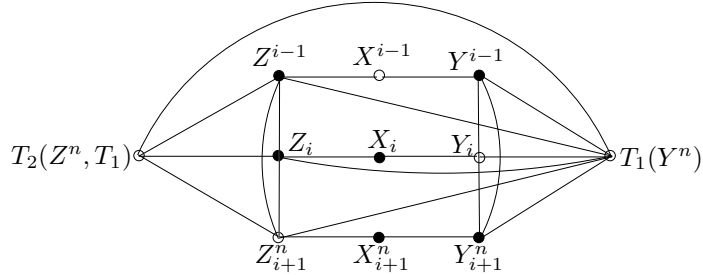}}

\caption{A graphical proof of the Markov chain
$Z^{i-1}-(T_1(Y^n),T_2(T_1,Z^n),X^{i-1},Z_{i}^n)-X_i$. The
undirected graph corresponds to the joint distribution
$p(x^{i-1},z^{i-1})p(y^{i-1}|x^{i-1})p(x_i,z_i)p(y_i|x_i)p(x_{i+1}^n,z_{i+1}^n)p(y_{i+1}^n|x_{i+1}^n)p(t_1|y^n)p(t_2|z^n,t_1)$.
The Markov chain holds since all paths from  $Z^{i-1}$ to $X_i$ pass
through $(T_1(Y^n),T_2(T_1,Z^n),X^{i-1},Z_{i}^n)$.
}\label{f_markov10}}\end{figure}

\begin{figure}[h!]{
\psfrag{X1}[][][1]{$X^{i-1}$} \psfrag{X2}[][][1]{$X_i$}
\psfrag{X3}[][][1]{$X_{i+1}^{n}$} \psfrag{Y1}[][][1]{$Y^{i-1}$}
\psfrag{Y2}[][][1]{$Y_i$} \psfrag{Z3}[][][1]{$Z_{i+1}^{n}$}
\psfrag{Z1}[][][1]{$Z^{i-1}$} \psfrag{Z2}[][][1]{$Z_i$}
\psfrag{Z3}[][][1]{$Z_{i+1}^n$} \psfrag{Y3}[][][1]{$Y_{i+1}^{n}$}
\psfrag{T'}[][][1]{$\;\;\;\;\;\;\;\;\;T_1(Y^n)$}
\psfrag{T}[][][1]{$T_3(X^n,T_1)\;\;\;\;\;\;\;\;\;\;\;\;\;\;$}
%{C:/cygwin/usr/X11R6/bin/home/rate_helper_side_simpl.eps}}
%\centerline{\includegraphics[width=7cm]{markov6.eps}}
%\centerline{\includegraphics[width=7cm]{c:/home/markov10.eps}}
\centerline{\includegraphics[width=7cm]{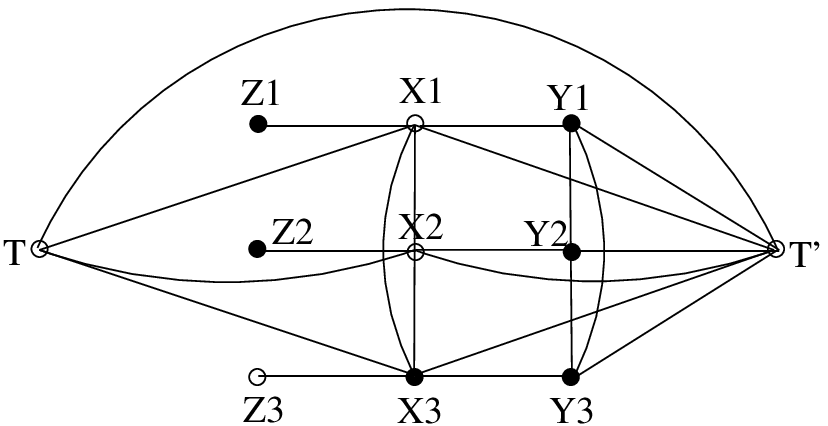}}

\caption{A graphical proof of the Markov chain
$Z_i-(U_i,W_i,X_i)-X_{i+1}^n$. The undirected graph corresponds to
the joint distribution
$p(x^{i-1},z^{i-1})p(y^{i-1}|x^{i-1})p(x_i,z_i)p(y_i|x_i)p(x_{i+1}^n,z_{i+1}^n)p(y_{i+1}^n|x_{i+1}^n)p(t_1|y^n)p(t_3|x^n,t_1)$.
The Markov chain holds since all paths from  $Z^{i}$ to $X_{i+1}^n$
pass through $(T_1(Y^n),T_3(T_1,X^n),X^{i},Z_{i+1}^n)$.
}\label{f_markov11}}\end{figure}

  Finally, let $Q$ be a random
variable independent of $X^n,Y^n,Z^n$, and uniformly distributed
over the set $\{1,2,3,..,n\}$. Define the random variables
$U\triangleq(Q,U_Q)$, $V\triangleq (Q,V_Q)$,  $W\triangleq (Q,W_Q)$,
$\hat X \triangleq\hat X_Q$, and $\hat Z \triangleq \hat Z_Q$. Then
(\ref{e_conv_two_way_i})-(\ref{e_cond1_two_way}) imply that the
equations that define $\mathcal R(D_x,D_z)$ i.e.,
(\ref{e_R1_twoway})-(\ref{e_def_dist}),  hold.

\QED

\section{ Two-way multi stage\label{s_two_way_multi}}

Here we consider the two-way multi-stage rate-distortion problem
with a helper. First, the helper sends a common message to both
users, and then users $X$ and $Z$ send to each other a total rate
$R_x$ and $R_z$, respectively, in $K$ rounds.  We use the definition
of two-way source coding as given in \cite{Kaspi85_two_way}, where
each user may transmit up to $K$ messages to the other user that
depends on the source and previous received messages. %\footnote{The
%case where a helper is absent and the communication between the
%users is via memoryless channels was recently solved by Maor and
%Merhav \cite{MaorM_merhav06_two_way} where they showed that a source
%channel separation theorem holds.}

Let $\mathcal M\ $ denote a set of positive integers $\{1,2,..,M\}$
and let $\mathcal M^K$ the collection of $K$ sets $\{\mathcal
M_1,\mathcal M_2,..., \mathcal M_K\}$.

\begin{figure}[h!]{
\psfrag{b1}[][][1]{$X$} \psfrag{box1}[][][1]{User X}
\psfrag{box3}[][][1]{Helper } \psfrag{t1}[][][1]{$R_y$}
\psfrag{box2}[][][1]{User Z} \psfrag{b3}[][][1]{$\hat X$}
\psfrag{b4}[][][1]{$\hat Z$} \psfrag{Y}[][][1]{$Y$}
\psfrag{Z}[][][1]{$Z$} \psfrag{t1}[][][1]{$R_{y}$}
\psfrag{t2}[][][1]{$R_{z,k}$} \psfrag{t3}[][][1]{$R_{x,k}$}
\psfrag{W}[][][1]{} \psfrag{X}[][][1]{} \psfrag{V}[][][1]{}
\psfrag{YZ}[][][1]{}\psfrag{D}[][][1]{}

%\centerline{\includegraphics[width=13cm]{c:/home/rate_dist_two_way.eps}}
\centerline{\includegraphics[width=13cm]{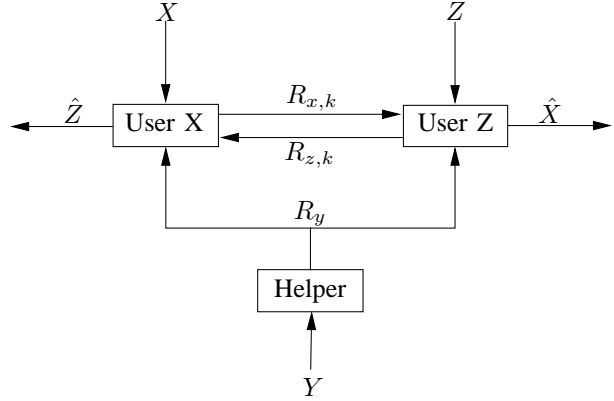}}
\caption{The two-way multi-stage with a helper. First Helper Y sends
a common message to User X and to User Z at rate $R_y$, and then we
have $K$ rounds where in each round $k\in\{1,...,K\}$ User Z sends a
message to User X at rate $R_{z,k}$, and User X sends a message to
User $Z$ at rate $R_{x,k}$. The limitation is on rate $R_y$ and on
the sum rates $R_x=\sum_{k=1}^K R_{x,k}$ and  $R_z=\sum_{k=1}^K
R_{z,k}$. % The goal is that User X will reconstruct the sequence
%$Z^n$ within a fidelity criterion
%$\mathbb{E}\left[\frac{1}{n}\sum_{i=1}^n d_z(Z_i,\hat
%Z_i)\right]\leq D_z$, and User Z will reconstruct the source $X^n$
%within a fidelity criterion $\frac{1}{n}\mathbb{E}\left[\sum_{i=1}^n
%d_x(X_i,\hat X_i)\right]\leq D_x.$
We assume that the side information $Y$ and the two sources $X,Z$
are i.i.d. and form the Markov chain $Y-X-Z$.}
\label{f_helper_two_way_multi} }\end{figure}

\begin{definition}\label{def_code}
An $(n,M_y,M_{x}^K,M_{z}^{K},D_x,D_z)$ code for two sources $X$ and
$Z$ with helper $Y$ consists of the encoders
\begin{eqnarray}
f_y&:& \mathcal Y^n \to \mathcal M_y \nonumber \\% \{1,2,...,M_1\} \nonumber \\
f_{z,k} &:& \mathcal Z^n \times \mathcal M^{k-1}_x \times \mathcal M_y \to \mathcal M_{z,k}, \quad k=1,2,...,K \nonumber \\
f_{x,k} &:& \mathcal X^n \times \mathcal M^{k}_z \times \mathcal M_y
\to \mathcal M_{x,k}, \quad k=1,2,...,K
\end{eqnarray}
%\end{equation}
and two decoders
\begin{eqnarray}
g_x &:&  \mathcal X^n \times \mathcal M_y \times \mathcal M^K_z
\to \hat{\cal Z}^n \nonumber \\
g_z &:&  \mathcal Z^n  \times \mathcal M_y \times \mathcal M^K_x
 \to \hat{\cal X}^n
\end{eqnarray}
such that
\begin{eqnarray}
\mathbb{E}\left[\sum_{i=1}^n d_x(X_i,\hat X_i)\right]&\leq& D_x,\nonumber \\
\mathbb{E}\left[\sum_{i=1}^n d_z(Z_i,\hat Z_i)\right]&\leq& D_z,
\end{eqnarray}
\end{definition}
The rate triple $(R_x,R_y,R_z)$ of the code is defined by
\begin{eqnarray}
R_y&=&\frac{1}{n}\log M_y;  \nonumber \\
R_x&=&\frac{1}{n}\sum_{i=1}^K \log M_{x,i};  \nonumber \\
R_z&=&\frac{1}{n}\sum_{i=1}^K \log M_{z,i};
\end{eqnarray}

Let us denote by ${\cal R}_K^O(D_x,D_z)$ the (operational)
achievable region of the multi-stage rate distortion with a helper,
i.e.,  the closure of the set of all triple rate $(R_x,R_y,R_z)$
that are achievable with a distortion pair $(D_x,D_z)$. Let
$\mathcal R_{K}(D_x,D_z)$ be the set of all triple rates
$(R_x,R_y,R_z)$ that satisfy
\begin{eqnarray}
R_y&\geq& I(U;Y),\label{e_multi_Ry} \\
R_z&\geq& \sum_{k=1}^{K} I(Z;V_k|X,U,V^{k-1},W^{k-1}), \label{e_multi_Rz}  \\
R_x&\geq& \sum_{k=1}^{K} I(X;W_k|Z,U,V^{k},W^{k-1}),
\label{e_multi_Rx}
\end{eqnarray}
for some auxiliary random variables $(U,V^K,W^k)$ that satisfy
\begin{equation}\label{e_markov_multi1}
U-Y-(X,Z),\\
\end{equation}
\begin{equation}\label{e_markov_multi2}
V_k-(Z,U,V^{k-1},W^{k-1})-(X,Y),\ \ k=1,2,...,K,\\
\end{equation}
\begin{equation}\label{e_markov_multi3}
W_k-(X,U,V^k,W^{k-1})-(Z,Y),\ \ k=1,2,...,K,
\end{equation}
\bea \mathbb{E}d_x(X,\hat X(U,W^K,Z))&\leq& D_x,\nonumber \\
\mathbb{E}d_z(Z,\hat Z(U,V^K,X))&\leq& D_z.\label{e_dist_multi} \eea

The Markov chain $Y-X-Z$ and the Markov chains given in
(\ref{e_markov_multi1})-(\ref{e_markov_multi3}) imply that the joint
distribution of $X,Y,Z,U,V^k,W^k$ is
 of the form $p(x,y)p(z|x)p(u|y)\prod_{k=1}^K
p(v_k|z,u,v^{k-1},w^{k-1})p(w_k|x,u,v^{k},w^{k-1}).$ Furthermore,
(\ref{e_multi_Rz}) and (\ref{e_multi_Rx}) can be written as
\begin{eqnarray}
R_z&\geq& I(Z;V^K,W^K|X,U),\label{e_multi_Rz2} \\
R_x&\geq& I(X;V^K,W^K|Z,U), \label{e_multi_Rx2}
\end{eqnarray}
due to the  the Markov chains $Z-(X,U,V^k,W^{k-1})-W_k$ and
$X-(Z,U,V^{k-1},W^{k-1})-V_k$.

\begin{lemma}
\begin{enumerate}
\item \label{lemma_properties_R_convex}
The region ${\cal R}_K(D_x,D_z)$ is convex
\item \label{lemma_properties_R_size}
To exhaust ${\cal R}_K(D_x,D_z)$, it is enough to restrict the
alphabet of $U$, $V$, and $W$ to satisfy
\begin{eqnarray}
|{\cal U}|&\leq& |{\cal Y}|+2K+1,\nonumber \\
|{\mathcal V_k}|&\leq& |{\cal Z}||{\cal U}||{\mathcal V^{k-1}}||{\mathcal W^{k-1}}|+2(K+1-k)+1,\;\;\; \text{for } k=1,..,K, \nonumber \\
|{\mathcal W_k}| &\leq& |{\cal X}||{\cal U}||{\mathcal
V^{k}}||{\mathcal W^{k-1}}|+2(K+1-k),\;\;\; \text{for } k=1,..,K.
\end{eqnarray}
\end{enumerate}
\end{lemma}
The proof of the lemma is analogous to the proof of
Lemma~\ref{lemma:properties_R} and therefore omitted.

\begin{theorem}\label{t_two_way_multi}
In the  two-way problem with $K$ stages of communication and a
helper, as depicted in Fig. \ref{f_helper_two_way_multi}, where
$Y-X-Z$,
\begin{equation}
\mathcal R_K^O(D_x,D_z)=\mathcal R_K(D_x,D_z).
\end{equation}
\end{theorem}

Theorem \ref{t_two_way_multi} is a generalization of Theorem
\ref{t_two_way} (equations (\ref{e_multi_Ry})-(\ref{e_dist_multi})
where $K=1$ are equivalent to
(\ref{e_R1_twoway})-(\ref{e_def_dist})) and its proof is a
straightforward extension. Here we explain only the extensions.

%Theorem \ref{t_two_way} is a special case of Theorem
%\ref{t_two_way_multi}, where $K=1$, and the proof of Theorem
%\ref{t_two_way_multi} is a straightforward extension of the proof of
%Theorem \ref{t_two_way}, which is given in Section
%\ref{s_proof_t_two_way}.

{\bf Sketch of achievability:} In the achievability proof of Theorem
\ref{t_two_way}, we generated the sequences $(U^n,V_1^n,W_1^n)$ that
are jointly typical with $X^n,Y^n,Z^n$. Using the same idea of
Wyner-Ziv coding we continue and generate at any stage
$k=1,2,...,K$, the sequence $V_k^n$ that is jointly typical with the
other sequences by transmitting a message at rate
$I(Z;V_k|X,U,V^{k-1},W^{k-1})$ from User Z to User X, and similarly
the sequence $W_k^n$ that is jointly typical with the other
sequences by transmitting a message at rate
$I(X;W_k|Z,U,V^{k},W^{k-1})$ from User X to User Z. In the final
stage, User X uses the sequences $(X^n,U^n,V_1^n,...,V_K^n)$ to
construct $\hat Z^n$ and, similarly, User Z uses the sequences
$(Z^n,U^n,W_1^n,...,W_K^n)$ to construct $\hat X^n$.

{\bf Sketch of Converse:} Assume that we have an
$(n,M_y,M_{x}^K,M_{z}^{K},D_x,D_z)$ code and we will show the
existence of a vector $(U,V^K,W^K,\hat{X},\hat{Z})$ that
satisfy~(\ref{e_multi_Ry})-(\ref{e_dist_multi}). Denote
$T_y=f_y(Y^n)$, $T_{z,k}=f_{z,k}(Z^n,T_y,T_x^{k-1})$, and
$T_{x,k}=f_{x,k}(X^n,T_y,T_z^k)$. Then the same arguments as  in
(\ref{e_R1_tow_way_con}) we obtain
\begin{eqnarray}\label{e_R1_tow_way_con_multi}
nR_y &\stackrel{}{\geq}&\sum_{i=1}^n
H(Y_i;X^{i-1},T_y,Z_{i+1}^n|Z_i)
%H(Y_i|Z_i)-H(Y_i|X^{i-1},T_y,Z_i^n),
\end{eqnarray}%
Then we have

\begin{eqnarray}
nR_z &\stackrel{}{\geq} & H(T_{z}^K)=\sum_{k=1}^K
H(T_{z,k}|T_{z}^{k-1}) {\geq} \sum_{k=1}^K
H(T_{z,k}|T_{z}^{k-1},T_{x}^{k-1}),  \label{e_HTz} \\
nR_x
&\stackrel{}{\geq} & H(T_{x}^K)=\sum_{k=1}^K
H(T_{x,k}|T_{x}^{k-1}){\geq} \sum_{k=1}^K
H(T_{x,k}|T_{x}^{k-1},T_{z}^{k})\label{e_HTx}.
\end{eqnarray}
Applying the same arguments as in (\ref{e_R2_tow_way_con}) and
(\ref{e_R3_tow_way_con}) on the terms in (\ref{e_HTz}) and
(\ref{e_HTx}), respectively, we obtain that

\begin{eqnarray}
H(T_{z,k}|T_{z}^{k-1},T_{x}^{k-1}) &\stackrel{}{\geq} &\sum_{i=1}^n I(Z_i;T_{z,k}|Z_{i+1}^n,X^{i},T_y,T_{z}^{k-1},T_{x}^{k-1}) \nonumber \\
H(T_{x,k}|T_{x}^{k-1},T_{z}^{k})  &\stackrel{}{\geq} &\sum_{i=1}^n
I(X_i;T_{x,k}|Z_{i}^n,X^{i-1},T_y,T_{z}^{k},T_{x}^{k-1}).
\end{eqnarray}

We define the auxiliary random variables as $U\triangleq
X^{Q-1},T_y,Z_{Q+1}^n$, $V_{k}=T_{z,k}$ and $W_{k}=T_{x,k}$, where
$Q$ is distributed uniformly on the integers $\{1,2,...,n\}$. \hfill
\QED

\section{Gaussian Case\label{s_gaussian_case}}
In this subsection we consider the Gaussian instance of the two way
setting with a helper as defined in Section \ref{s_definition} and
explicitly express the region for a mean square error distortion (we
 also note that the multi
 stage option does not increase the rate region for this case).

\begin{figure}[h!]{
\psfrag{b1}[][][1]{$X=Z+A$} \psfrag{box1}[][][1]{User X}
\psfrag{box3}[][][1]{Helper } \psfrag{t1}[][][1]{$R_y$}
\psfrag{box2}[][][1]{User Z} \psfrag{b3}[][][1]{$\hat X$}
\psfrag{b4}[][][1]{$\hat Z$} \psfrag{Y}[][][1]{$Y=Z+A+B$}
\psfrag{Z}[][][1]{$Z$} \psfrag{t1}[][][1]{$R_{y}$}
\psfrag{t2}[][][1]{$R_{z}$} \psfrag{t3}[][][1]{$R_{x}$}
\psfrag{W}[][][1]{$A \sim \text{N}(0,\sigma^2_A),$}
\psfrag{X}[][][1]{$B \sim \text{N}(0,\sigma^2_B),$}
\psfrag{V}[][][1]{$Z \sim \text{N}(0,\sigma^2_Z),$}
\psfrag{YZ}[][][1]{$A\perp B \perp Z,$}
\psfrag{D}[][][1]{square-error distortion}
%\centerline{\includegraphics[width=13cm]{c:/home/rate_dist_two_way.eps}}
\centerline{\includegraphics[width=13cm]{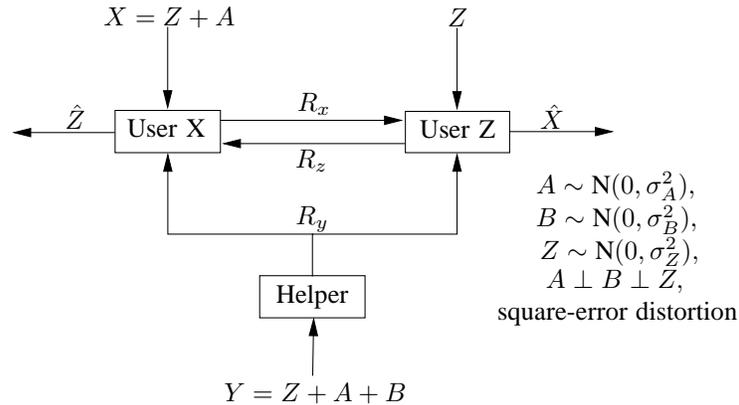}}
\caption{The Gaussian two-way  with a helper. The side information
$Y$ and the two sources $X,Z$ are i.i.d., jointly  Gaussian and form
the Markov chain $Y-X-Z$.  The distortion is the square error, i.e.,
$d_x(X^n,\hat X^n)=\frac{1}{n}\sum_{i=1}^n (X_i - \hat X_i)^2$ and
$d_z(Z^n,\hat Z^n)=\frac{1}{n}\sum_{i=1}^n (Z_i - \hat Z_i)^2$.}
\label{f_helper_two_way_gaussin} }\end{figure}

%In this subsection we consider the Gaussian instance that
%corresponds to Theorem \ref{t_rate_dis_helper_side} (see Fig.
%\ref{f_gaussian_side}).
Since $X,Y,Z$ form the Markov chain $Y-X-Z$,
we assume, without loss of generality, that $X=Z+A$ and $Y=Z+A+B$,
where the random variables $(A,B,Z)$ are zero-mean Gaussian and
independent of each other, where $\mathbb E[A^2]=\sigma_A^2$,
$\mathbb E[B^2]=\sigma_B^2$ and $\mathbb E[Z^2]=\sigma_Z^2$.
\begin{corollary}\label{t_gaussian_region_side}
The achievable rate region of the problem illustrated in Fig.
\ref{f_helper_two_way_gaussin} is
\begin{eqnarray}
R_z&\geq&  \frac{1}{2}\log
\frac{\sigma_A^2\sigma_Z^2}{D_z(\sigma_A^2+\sigma_Z^2)},\label{e_gayussian_Rz}
\\
R_x&\geq&  \frac{1}{2}\log
\frac{\sigma_A^2\left(\sigma_B^2+\sigma_A^2
2^{-2R_y}\right)}{D_x(\sigma_A^2+\sigma_B^2)}.
\end{eqnarray}
\end{corollary}
\begin{proof}
The converse and achievability of (\ref{e_gayussian_Rz}) follows
from the Gaussian Wyner-Ziv coding {\cite{Wyner78_gaussian_WZ}}
result, which states that the achievable rate for the Gaussian
Wyner-Ziv setting is the same as the case where the side information
is known to the encoder and decoder. Furthermore, because of the
Markov chain $Z-X-Y$, the rate $R_y$ does not have any influence on
$R_z$, since this rate is the achievable rate even if $Y$ is known
to both users. The achievability and the converse for $R_x$ is given
in the following corollary.
\end{proof}

\begin{figure}[h!]{
 \psfrag{box1}[][][1]{$X=Z+A$}
\psfrag{box3}[][][1]{$Y=Z+A+B$} \psfrag{Side}[][][1]{$Z\;$}

 \psfrag{a2}[][][1]{$T_x\in 2^{nR_x}$}
 \psfrag{t1}[][][1]{$$}
\psfrag{A1}[][][1]{$$} \psfrag{A2}[][][1]{$$}
 \psfrag{box2}[][][1]{$\; \hat X^n$}
\psfrag{b3}[][][1]{$\hat X$} \psfrag{a3}[][][1]{}
\psfrag{W}[][][1]{$A \sim \text{N}(0,\sigma^2_A),$}
\psfrag{Y}[][][1]{$B \sim \text{N}(0,\sigma^2_B),$}
\psfrag{Z}[][][1]{$Z \sim \text{N}(0,\sigma^2_Z),$}
\psfrag{YZ}[][][1]{$A\perp B \perp Z,$}
\psfrag{D}[][][1]{square-error distortion}
\psfrag{t2}[][][1]{$T_y\in 2^{nR_y}$}
\centerline{\includegraphics[width=15cm]{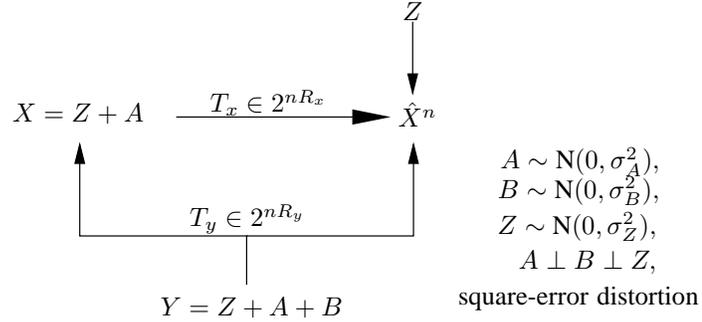}}
%\centerline{\includegraphics[width=15cm]
%{C:/cygwin/usr/X11R6/bin/home/rate_helper_side_simple.eps}}
\caption{Gaussian case: the zero-mean Gaussian random variables
$A,B,Z$ are i.i.d. and independent of each other. Their variances
are $\sigma^2_A$, $\sigma^2_B$ and $\sigma^2_Z$, respectively. The
source $X$ and the helper $Y$ satisfy $X=A + Z$ and $Y=Z+A+B$. The
distortion is the square error, i.e., $d(X^n,\hat
X^n)=\frac{1}{n}\sum_{i=1}^n (X_i - \hat X_i)^2$.}
\label{f_gaussian_side} }\end{figure}
%
%The following corollary establishes the rate region of the
%Gaussian case.
\begin{corollary}\label{t_gaussian_region_side}
The achievable rate region of the problem illustrated in Fig.
\ref{f_gaussian_side} is
\begin{eqnarray}\label{e_gaussian_helper}
R&\geq& \frac{1}{2}\log
\frac{\sigma_A^2\left(1-\frac{\sigma_A^2}{\sigma_A^2+\sigma_B^2}
(1-2^{-2R_y})\right)}{D}%\nonumber \\
%&=& \frac{1}{2}\log
%\frac{\sigma_A^2\left(\sigma_B^2-\sigma_A^2(1-2^{-2R_1})\right)}{D(\sigma_A^2+\sigma_B^2)}
\end{eqnarray}
\end{corollary}

It is interesting to note that the rate region does not depend on
$\sigma_Z^2$. Furthermore, we show in the proof that for the
Gaussian case the rate region is the same as when $Z$ is known to
the source $X$ and the helper $Y$.

{\it Proof of Corollary \ref{t_gaussian_region_side}:}

{\bf Converse}: Assume that both encoders observe $Z^n$. Without
loss of generality, the encoders can subtract $Z$ from $X$ and $Y$;
hence the problem is equivalent to new rate distortion problem with
a helper, where the source is $A$ and the helper is $A+B$. Now using
the result for the Gaussian case from~\cite{Vasudevan07_helper},
adapted to our notation, we obtain (\ref{e_gaussian_helper}).
%\begin{equation}\label{e_gaussian_region_sigma_side}
%R\geq \frac{1}{2} \log
%\frac{\sigma_A^2\left(1-\frac{\sigma_A^2}{\sigma_A^2+\sigma_B^2}
%(1-2^{-2R'})\right)}{D}.
%\end{equation}
{\bf Achievability:} Before proving the direct-part of Corollary
\ref{t_gaussian_region_side}, we establish the following lemma which
is proved in Appendix \ref{app_proof_wyner_siz_side}.
%\ref{app_lema_wyner_ziv_side_gaussian}.
\begin{lemma}\label{l_wyner_ziv_side_gaussian}
{\it Gaussian Wyner-Ziv rate-distortion problem with additional
side information known to the encoder and decoder.} Let $(X,W,Z)$
be jointly Gaussian. Consider the Wyner-Ziv rate distortion
problem where the source $X$ is to be compressed with quadratic
distortion measure, $W$ is available at the encoder and decoder,
and $Z$ is available only at the decoder. The rate-distortion
region for this problem is given by
\begin{equation}
R(D)=\frac{1}{2}\log \frac{\sigma^2_{X|W,Z}}{D},
\end{equation}
where $\sigma^2_{X|W,Z}=\mathbb E[(X-\mathbb E[X|W,Z])^2]$, i.e.,
the minimum square error of estimating $X$ from $(W,Z)$.
\end{lemma}

Let $V=A+B+Z+D$, where  $D\sim \text{N}(0,\sigma_D^2)$ and is
independent of $(A,B,Z)$. Clearly, we have $V-Y-X-Z$. Now, let us
generate $V$ at the source-encoder and at the decoder using the
achievability scheme of Wyner \cite{Wyner78_gaussian_WZ}. Since
$I(V;Z)\leq I(V;X)$ a rate $R'=I(V;Y)-I(V;Z)$ would suffice, and
it may be expressed as follows:
\begin{eqnarray}
R'&\stackrel{}{=}&I(V;Y|Z)\nonumber \\
&\stackrel{}{=}&h(V|Z)-h(V|Y)\nonumber \\
&\stackrel{}{=}&\frac{1}{2}\log
\frac{\sigma_A^2+\sigma_B^2+\sigma_D^2}{\sigma_D^2},
\end{eqnarray}
and this implies that
\begin{equation}\label{e_sigma_D}
\sigma_D^2=\frac{\sigma_A^2+\sigma_B^2}{2^{2R'}-1}.
\end{equation}
Now, we invoke  Lemma \ref{l_wyner_ziv_side_gaussian}, where $V$
is the side information known both to the encoder and decoder;
hence a rate that satisfies the following inequality achieves a
distortion $D$;
\begin{eqnarray}
R&\stackrel{}{\geq}&\frac{1}{2}\log
\frac{\sigma^2_{X|V,Z}}{D}\nonumber \\
&\stackrel{}{=}&\frac{1}{2}\log
\frac{\sigma^2_{A}}{D}\left(1-\frac{\sigma_A^2}{\sigma_A^2+\sigma_B^2+\sigma_D^2}\right)
\end{eqnarray}
Finally, by replacing $\sigma_D^2$ with the identity in
(\ref{e_sigma_D}) we obtain (\ref{e_gaussian_helper}).
%\ref{e_sigma_D} &\stackrel{(a)}{=}&\frac{1}{2}\log
%\frac{\sigma^2_{A}}{D}\left(1-\frac{\sigma_A^2}{\sigma_A^2+\sigma_B^2}\left
%(1-2^{-2R'}\right )\right).
%\end{eqnarray}
%Step (a) is due to \ref{e_sigma_D}.
\hfill\QED

\section{Further results on Wyner-Ziv with a helper where $Y-X-Z$\label{s_wyner_ziv_helper_further}}
In this section we investigate two properties of the rate-region of
the Wyner-Ziv setting ( Fig. \ref{f_helper_side_again}) with a
Markov form $Y-X-Z$. First, we investigate the tradeoff between the
rate sent by the helper and the rate sent by the source and roughly
speaking we conclude that a bit from the source is  more ``valuable"
than a bit from the helper. Second, we examine the case where the
helper has the freedom to send different messages, at different
rates, to the encoder and the decoder. We show that ``more help'' to
the encoder than to the decoder does not yield any performance gain
and that in such cases the freedom to send different messages to the
encoder and the decoder yields no gain over the case of a common
message. Further, in this setting of different messages, the rate to
the encoder can be strictly less than that to the decoder with no
performance loss.

\begin{figure}[h!]{
\psfrag{b1}[][][1]{$X$} \psfrag{box1}[][][1]{Encoder}
\psfrag{box3}[][][1]{Helper\ }
 %\psfrag{a2}[][][1]{$T_1(X^n,T_2)$}
 %\psfrag{t1}[][][1]{$\in 2^{nR_1}$}
%\psfrag{b3}[][][1]{$\;\hat X^n(T_1,T_2)$}
%\psfrag{t2}[][][1]{$T_2(Y^n)\in 2^{nR_2}$}
 \psfrag{a2}[][][1]{$R$}
 \psfrag{t1}[][][1]{$$}
\psfrag{A1}[][][1]{$$} \psfrag{A2}[][][1]{$$}
 \psfrag{box2}[][][1]{Decoder}
\psfrag{b3}[][][1]{$\hat X$} \psfrag{a3}[][][1]{}
\psfrag{Y}[][][1]{$Y$} \psfrag{Z}[][][1]{$Z$}
\psfrag{t2}[][][1]{$R_1$}
%\centerline{\includegraphics[width=15cm]
%{C:/cygwin/usr/X11R6/bin/home/rate_distortion_helper_side.eps}}
\centerline{\includegraphics[width=13cm]{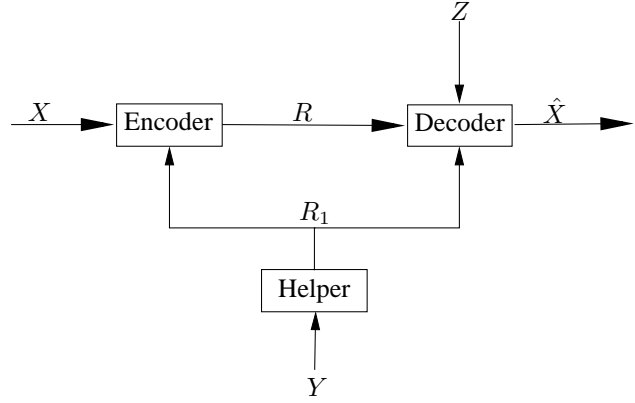}}
\caption{Wyner-Ziv problem with a helper where the Markov chain
$Y-X-Z$ holds.} \label{f_helper_side_again} }\end{figure}

\subsection{A bit from the source-encoder vs. a bit from the helper
\label{s_bit_from_source_encoder_vs_helper}} Assume that we have a
sequence of $(n,2^{nR},2^{nR_1})$ codes that achieves a distortion
$D$, such that the triple $(R,R_1,D)$ is on the border of the region
$\mathcal R_{Y-X-Z}(D)$ (recall the definition of $\mathcal
R_{Y-X-Z}(D)$ in
(\ref{e_R_R_1_ineq_helper_side_2})-(\ref{e_p_xyuv_decompose})). Now,
suppose that the helper is allowed to increase the rate by an amount
$\Delta'>0$ to $R_1+\Delta'$; to what rate $R-\Delta$ can the
source-encoder reduce its rate and achieve the same distortion $D$?

Despite the fact that the additional rate $\Delta'$ is transmitted
both to the decoder and encoder, we show that always
$\Delta\leq\Delta'$. Let us denote by $R(R_1)$  the boundary of the
region $\mathcal R_{Y-X-Z}( D)$  for a fixed $D$. We formally show
that $\Delta\leq\Delta'$ by proving that the slope of the curve
$R(R_1)$ is always less than 1.  The proof uses similar technique as
in \cite{Steinberg08RateLimitedAtDecoder}.
\begin{lemma}
For any $X-Y-Z$, $D$, and  $R_1$, the subgradients of the curve
$R(R_1)$ are less than 1.
\end{lemma}
\begin{proof}
Since $\mathcal R_{Y-X-Z}(D)$ is a convex set,  $R(R_1)$  is a
convex function. Furthermore, $R(R_1)$ is non increasing in $R_1$.
Now, let us define $J^*(\lambda)$ as
\begin{equation}
J^*(\lambda)=\min_{p(x,y,z,u,w)\in \mathcal P} I(X;W|U,Z)+\lambda
I(Y;U|Z),
\end{equation}
where $\mathcal P$ is the set of distributions satisfying $ p(
x,y,z,u,w,\hat x)=p(x,y)p(z|y)p(u|y)p(w|u,x)p(\hat x|u,w,z), \quad
\mathbb{E}d(X,\hat X)\leq D.$ The line $J^*(\lambda)=R+\lambda R$ is
a support line of $R(R_1)$, and therefore, $\lambda$ is a
subgradient. The value $J^*(\lambda)$ is the  intersection  between
the support line with slope $-\lambda$ and the axis $R$, as shown in
Fig. \ref{f_support}. Because of the convexity and the monotonicity
of $R(R_1)$, $J^*(\lambda)$ is upper-bounded by  $R(0)$, i.e.,
\begin{equation}\label{e_J_upper_bound}
J^*(\lambda)\leq \min_{p(\hat x, x,y,z,u,w)\in \mathcal P}
R(0)=\min_{p(\hat x, x,y,z,w)\in \mathcal P_{WZ}} I(X;W|Z),
\end{equation}
where $\mathcal P_{WZ}$ is the set of distributions that satisfies $
p(\hat x, x,z,w)=p(x)p(z|x)p(w|x)p(\hat x|w,z), \quad
\mathbb{E}d(X,\hat X)\leq D.$
\begin{figure}[h!]{
 \psfrag{J}[][][0.8]{$J^*(\lambda)\;\;\;\;\;$}
 \psfrag{R1}[][][1]{$R$}
 \psfrag{R2}[][][1]{$R_1$}
 \psfrag{s}[][][0.7][-29]{support line with slope $-\lambda$}
 \psfrag{r2}[][][0.8]{$\min_{p(\hat x|x)} I(X;\hat X)
\to\;\;\;\;\;\;\;\;\;\;\;\;\;\;\;\;\;\;\;\;\;\;\;\;\;\;\;\;\;\;$}
 \psfrag{r1}[][][0.8]{$\min_{p(\hat x|x,y)}I(X;\hat X|Y)
\to\;\;\;\;\;\;\;\;\;\;\;\;\;\;\;\;\;\;\;\;\;\;\;\;\;\;\;\;\;\;\;\;\;\;\;\;$}

\centerline{\includegraphics[width=6cm]{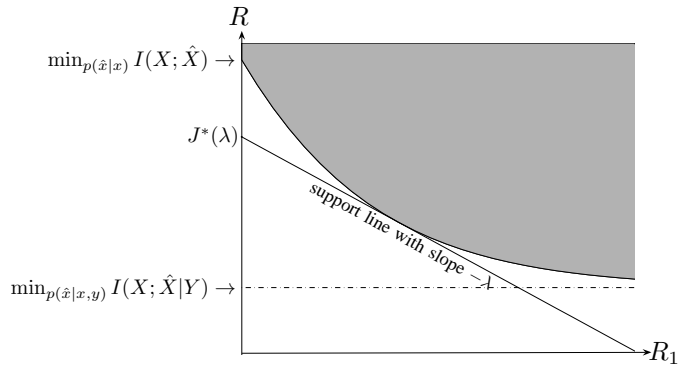}}
%\centerline{\includegraphics[width=6cm]{matlab/f_support2.eps}}
\caption{A support line of $R(R_1)$ with a slope $-\lambda$.
$J*(\lambda)$ is the intersection of the support line with the $R$
axis.} \label{f_support} }\end{figure} In addition, we observe that
%\begin{eqnarray}
% J^*(0)&=&\min_{p(x,y,u,\hat x)\in \mathcal P} I(X;\hat X|U)\nonumber \\
%&=& \min_{p(\hat x|x,y)}I(X;\hat X|Y),
%\end{eqnarray}
%since we can choose $U=Y$. Furthermore
\begin{eqnarray}\label{e_J1}
 J^*(1)&=&\min_{p( x,y,z, u,w,\hat x)\in \mathcal P} I(X;W|U,Z)+I(Y;U|Z)\nonumber \\
 &\stackrel{\mbox{(a)}}{=}&\min_{p( x,y,z, u,w,\hat x)\in \mathcal P}
I(X,Y;W,U|Z)\nonumber \\
&\stackrel{}{\geq}& \min_{p( x,y,z, u,w,\hat x)\in \mathcal P}
I(X;W|Z),\nonumber \\
&\stackrel{}{=}& \min_{p(\hat x, x,y,z,w)\in \mathcal P_{WZ}}
I(X;W|Z),
\end{eqnarray}
where step (a) is due to the Markov chains $U-Y-(Z,X)$ and
$W-(U,X)-(Y,Z)$. Combining (\ref{e_J_upper_bound}) and (\ref{e_J1}),
we conclude that for any subgradient  $-\lambda$, $J^*(\lambda)\leq
J^*(1)$. Since $J^*(\lambda)$ is increasing in $\lambda$, we
conclude that $\lambda\leq 1$. %Finally, since $R(R_1)$ is
%non-increasing in $R_1$ we also conclude that $\lambda\geq 0$
\end{proof}
An alternative and equivalent proof would be to claim that, since
$R(R_1)$ is a convex and  non increasing function, $
\frac{\Delta}{\Delta'}\stackrel{}{\leq}
\left|\frac{dR}{dR_1}\right|_{R_1=0}$,  and then to claim that the
largest slope at $R_1=0$
 is when $Y=X$, which is 1. For the Gaussian case,  the
derivative may be calculated explicitly from
(\ref{e_gaussian_helper}), in particular for  $R_1=0$, and we obtain
\begin{eqnarray}\label{e_derivative_Gau}
\Delta\stackrel{}{\leq}\frac{\sigma_A^2}{\sigma_A^2+\sigma_B^2}\Delta'.
\end{eqnarray}
%and equality holds if $R_1=0$ and $\Delta'\to 0$.

\subsection{The case of independent rates} \label{subsec:WZHI} In
this subsection we treat the rate distortion scenario where side
information from the helper is encoded using two different messages,
possibly at different rates, one to the encoder and one to the
decoder, as shown in Fig.~\ref{fig:Helper_independent_rates}. The
complete characterization of achievable rates for this scenario is
still an open problem. However, the solution that is given in
previous sections, where there is one message known both to the
encoder and decoder, provides us insight that allows us to solve
several cases of the problem shown here.
\begin{figure}[t]{
\psfrag{X}[bl][bl][0.8]{$X^n$} \psfrag{E}[][][0.8]{Encoder}
\psfrag{H}[][][0.8]{Helper\ }
 \psfrag{A}[][][0.8]{$T$, rate $R$}
%  \psfrag{D}[][][0.8]{Decoder 1}
\psfrag{J}[][][0.8]{$T_d$, rate $R_d$}
%\psfrag{B}[][][0.8]{$T_2$, rate $R_2$}
\psfrag{C}[][][0.8]{Decoder} \psfrag{F}[bl][bl][0.8]{$\hat{X}^n$}
\psfrag{H}[][][0.8]{Helper} \psfrag{I}[bl][bl][0.8]{$Y^n$}
\psfrag{Z}[bl][bl][0.8]{$Z^n$} \psfrag{G}[][][0.8]{$T_e$, rate
$R_e$}\centerline{\includegraphics[viewport=80 578 508 796, clip,
     width=0.5\textwidth]{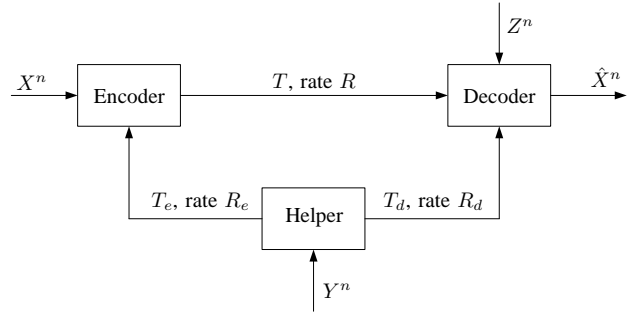}}
\caption{The rate distortion problem with decoder side
information, and independent helper rates. We assume the Markov
relation $Y - X - Z$} \label{fig:Helper_independent_rates} }
\end{figure}
We start with the definition of the general case. %It follows quite
%closely Definition~\ref{def_code_helper_side}, except that there
%are three rates involved.
\begin{definition}
\label{def_code_general} An $(n,M,M_e,M_d,D)$ code for source $X$
with side information $Y$ and different helper messages to the
encoder and decoder, consists of three encoders
\begin{eqnarray}
f_e&:& \mathcal Y^n \to \{1,2,...,M_e\} \nonumber \\
f_d&:& \mathcal Y^n \to \{1,2,...,M_d\} \nonumber \\
f &:& \mathcal X^n \times \{1,2,...,M_e\} \to \{1,2,...,M\} \nonumber \\
\end{eqnarray}
%\end{equation}
and a decoder
\begin{eqnarray}
g &:&  \{1,2,...,M\} \times \{1,2,...,M_d\} \to \hat{\cal X}^n \nonumber \\
\end{eqnarray}
such that
\begin{equation}
\mathbb{E}d(X^n,\hat X^n)\leq D.
\end{equation}
\end{definition}
To avoid cumbersome statements, we will not repeat in the sequel
the words ``... different helper messages to the encoder and
decoder," as this is the topic of this section, and should be
clear from the context. The rate pair $(R,R_e,R_d)$ of the
$(n,M,M_e,M_d,D)$ code is
\begin{eqnarray}
R&=&\frac{1}{n}\log M\nonumber \\
R_e&=&\frac{1}{n}\log M_e\nonumber\\
R_d&=&\frac{1}{n}\log M_d
\end{eqnarray}
\begin{definition}
\label{def_achievable rates_gen} Given a distortion $D$, a rate
triple $(R,R_e,R_d)$ is said to be {\it achievable} if for any
$\delta>0$, and sufficiently large $n$, there exists an
$(n,2^{n(R+\delta)},2^{n(R_e+\delta)},2^{n(R_d+\delta)},D+\delta)$
code for the source $X$ with side information $Y$.
\end{definition}
\begin{definition}
\label{def_the achievable_region_gen} {\it The (operational)
achievable region} $\mathcal R_g^O(D)$ of rate distortion with a
helper known at the encoder and decoder is the closure of the set of
all achievable rate triples at distortion $D$.
\end{definition}
Denote by ${\cal R}_g^O(R_e,R_d,D)$ the \emph{section} of ${\cal
R}_g^O(D)$ at helper rates $(R_e,R_d)$. That is, \bea {\cal
R}_g^O(R_e,R_d,D) &=& \left\{R:\ \ (R,R_e,R_d)\
\mbox{ are achievable with distortion $D$} \right\} % \nonumber\\
%& & \mbox{}\hspace*{0.2cm} \left.
%   \mbox{}  \right\}, \label{eq:R_g_section}
\eea and similarly, denote by $\mathcal R(R_1,D)$ the section of the
region ${\mathcal R_{Y-X-Z}}(D)$, defined
in~(\ref{e_R_R_1_ineq_helper_side_2})-(\ref{eq:dist2}) at helper
rate $R_1$. Recall that, according to
Theorem~\ref{t_rate_dis_helper_side}, ${\mathcal R}(R_1,D)$ consists
of all achievable source coding rates when the helper sends common
messages to the source encoder and destination at rate $R_1$. The
main result of this section is the following.
\begin{theorem}
\label{theo:R_g} For any $R_e\geq R_d$,
\begin{equation}
{\cal R}_g^O(R_e,R_d,D) = {\mathcal R}(R_d,D)
\end{equation}
\end{theorem}
Theorem~\ref{theo:R_g} has interesting implications on the coding
strategy taken by the helper. It says that no gain in performance
can be achieved if the source encoder gets ``more help" than the
decoder at the destination (i.e., if $R_e>R_d$), and thus we may
restrict $R_e$ to be no higher than $R_d$. Moreover, in those
cases where $R_e=R_d$, optimal performance is achieved when the
helper sends to the encoder and decoder exactly the same message.
The proof of this statement uses operational arguments.

\emph{Proof of Theorem~\ref{theo:R_g}:} Clearly, the claim is
proved once we show the statement for $R_e=H(Y)$. In this
situation, we can equally well assume that the encoder has full
access to $Y$. Thus, fix a general scheme like in
Definition~\ref{def_code_general} with $R_e=H(Y)$. The encoder is
a function of the form $f(X^n,Y^n)$. Define $T_2=f_d(Y^n)$. The
Markov chain $Z - X - Y$ implies that $Z^n - (X^n,T_2) - Y^n$ also
forms a Markov chain. This implies, in turn that there exists a
function $\phi$ and a random variable $W$, uniformly distributed
in $[0,1]$ and independent of $(X^n,T_2,Z^n)$, such that \be Y^n =
\phi(X^n,T_2,W). \label{eq:def_phi} \ee Thus the source encoder
operation can be written as \bea
f(X^n,Y^n) &=& f(X^n,\phi(X^n,T_2,W))\nonumber\\
&\stackrel{\triangle}{=}&
\tilde{f}(X^n,T_2,W)\label{eq:def_tilde_f} \eea implying, in turn,
that the distortion of this scheme can be expressed as \bea
\lefteqn{ \mathbb{E} d(X^n,\hat{X}^n) = \mathbb{E}\left[
d(X^n,\hat{X}^n(\tilde{f}(X^n,T_2,W),T_2,Z^n))\right]}\nonumber\\
&\stackrel{\mbox{(a)}}{=}&\int_0^1 \mathbb{E}\left[
d(X^n,\hat{X}^n(\tilde{f}(X^n,T_2,w),T_2,Z^n))\right] dw \nonumber\\
&\stackrel{\mbox{(b)}}{=}& \int_0^1 \mathbb{E}\left[
d(X^n,\hat{X}^n(f^w(X^n,T_2),T_2,Z^n))\right] dw
\label{eq:distortion_tilde_f} \eea where (a) holds since $W$ is
independent of $(X^n,T_2,Z^n)$, and (b) by defining \be f^w(X^n,T_2)
= \tilde{f}(X^n,T_2,w). \label{eq:def_fw} \ee Note that for a given
$w$, the function $f^w$ is of the form of encoding functions where
the helper sends one message to the encoder and decoder. Therefore
we conclude that anything achievable with a scheme from
Definition~\ref{def_code_general}, is achievable by time-sharing
where the helper sends one message to the encoder and decoder.
\hfill\QED

The statement of Theorem~\ref{theo:R_g} can be extended to rates
$R_e$ slightly lower than $R_d$. This extension is based on the
simple observation that the source encoder knows $X$, which can
serve as side information in decoding the message sent by the
helper. Therefore, any message $T_2$ sent to the source decoder
can undergo a stage of binning with respect to $X$. As an extreme
example, consider the case where $R_e\geq H(Y|X)$. The source
encoder can fully  recover $Y$, hence there is no advantage in
transmitting to  the encoder at  rates higher than $H(Y|X)$; the
decoder, on the other hand, can benefit from rates in the region
$H(Y|X)<R_d<H(Y|Z)$. This rate interval is not empty due to the
Markov chain $Y - X - Z$. These observations are summarized in the
next theorem.
\begin{theorem}
\label{theo:R_g2}

\noindent
\begin{enumerate}
\item
\label{part:1_R_g2} Let $(U,V)$ achieve a point $(R,R')$ in ${\cal
R}_{Y-X-Z}(D)$, i.e., \bea
R &=& I(X;U|V,Z)\nonumber\\
R'&=& I(Y;V|Z)=I(V;Y)-I(V;Z) \label{eq:R_g2_1}\\
D &\geq& \mathbb{E}d(X,\hat{X}(U,V,Z)), \label{eq:R_g2_2}\\
 & & \ \ \ \ \ V - Y - X - Z.
\eea Then $(R,R_e,R')\in {\cal R}_g^O(D)$ for every $R_e$
satisfying \bea
 R_e &\geq& I(V;Y|Z) - I(V;X|Z)\nonumber\\
   & = &I(V;Y) - I(V;X). \label{eq:R_g2_5}
\eea
\item
Let $(R,R')$ be an outer point of ${\cal R}_{Y-X-Z}(D)$. That is, \be
(R,R')\not\in {\cal R}_{Y-X-Z}(D). \ee Then $(R,R_e,R')$ is an outer
point of ${\cal R}_g^O(D)$ for any $R_e$, i.e., \be
(R,R_e,R')\not\in {\cal R}_g^O(D)\ \ \forall\ R_e. \ee
\end{enumerate}
\end{theorem}
The proof of Part 1 is based on binning, as described above. In
particular, observe that $R_e$ given in~(\ref{eq:R_g2_5}) is lower
than $R'$ of~(\ref{eq:R_g2_1}) due to the Markov chain $V - Y - X
- Z$. Part 2 is a partial converse, and is a direct consequence of
Theorem~\ref{theo:R_g}. The details, being straightforward, are
omitted.

\appendices

\section{Proof of the the technique for verifying Markov relations \label{app_proof_technique}}
\emph{Proof} First let us prove that three random variables $X,Y,Z$,
with a joint distribution of the form
\begin{equation}
p(x,y,z)=f(x,y)f(y,z),
\end{equation}
satisfy the Markov chain $Y-X-Z$. Consider,
\begin{equation}
p(z|y,x)=\frac{f(x,y)f(y,z)}{f(x,y)\left(\sum_{z} f(y,z)\right)
}=\frac{ f(y,z)}{\sum_{z} f(y,z)},
\end{equation}
and since the expression does not include the argument $x$ we
conclude that $p(z|y,x)=p(z|y)$.

For the more general case, we first extend the sets $X_{\mathcal
G_{1}}$ $X_{\mathcal G_{3}}$. We start by defining
$\overline{\mathcal G}_{1}={\mathcal G}_{1}$ and $\overline{\mathcal
G}_{3}={\mathcal G}_{3}$, and then we  add to $X_{\overline{\mathcal
G}_{1}}$ and to $X_{\overline{\mathcal G}_{3}}$ all their neighbors
 that are not in $X_{\mathcal G_{2}}$ (a neighbor to a group  is a
node that is connected by one edge to the an element in the group).
We repeat this procedure till there are no more nodes to add to
$X_{\overline{\mathcal G}_{1}}$ or $X_{\overline {\mathcal G}_{3}}$.
Note that since there are no paths from $X_{\mathcal G_{1}}$ to
$X_{\mathcal G_{3}}$ that do not pass through $X_{\mathcal G_{2}}$,
then a node can not be added to both sets $X_{\overline{\mathcal
G}_{1}}$ and $X_{\overline{\mathcal G}_{3}}$. The set of nodes that
are not in $(X_{\overline{\mathcal G}_{1}}, X_{{\mathcal
G}_{2}},X_{\overline{\mathcal G}_{3}})$ is denoted as $X_{{\mathcal
G}_{0}}$.

The sets $X_{{\mathcal G}_{0}}$ and $X_{\overline{\mathcal G}_{1}}$
and $X_{\overline{\mathcal G}_{3}}$ are connected only to
$X_{{\mathcal G}_{2}}$ and not to each other, hence the joint
distribution of $(X_{{\mathcal G}_{0}},X_{\overline{\mathcal
G}_{1}}, X_{{\mathcal G}_{2}},X_{\overline{\mathcal G}_{3}})$ is of
the following form
\begin{equation}
p(X_{{\mathcal G}_{0}},X_{\overline{\mathcal G}_{1}}, X_{{\mathcal
G}_{2}}, X_{\overline{\mathcal G}_{1}} )=f(X_{{\mathcal
G}_{0}},X_{{\mathcal G}_{2}})f(X_{\overline{\mathcal
G}_{1}},X_{{\mathcal G}_{2}})f(X_{\overline{\mathcal
G}_{3}},X_{{\mathcal G}_{2}}).
\end{equation}
By marginalizing over $X_{{\mathcal G}_{0}}$ and using the claim
introduced in the first sentence of the proof we obtain the Markov
chain $X_{\overline{\mathcal G}_{1}}-X_{{\mathcal
G}_{2}}-X_{\overline{\mathcal G}_{3}}$, whcih implies $X_{{\mathcal
G}_{1}}-X_{{\mathcal G}_{2}}-X_{{\mathcal G}_{3}}$.
 \hfill \QED

\section{Proof of Lemma \ref{lemma:properties_R}\label{app_lemma_properties}}
\begin{proof}
To prove Part~\ref{lemma_properties_R_convex}, let $Q$ be a time
sharing random variable, independent of the source triple
$(X,Y,Z)$. Note that
\begin{eqnarray}
 I(Y;U|Z,Q)&\stackrel{(a)}{=}&I(Y;U,Q|Z)=I(Y;\tilde{U}|Z),  \nonumber \\
I(Z;V|U,X,Q)&=&I(Z;V|\tilde U,X), \nonumber \\
I(X;W|U,V,Z,Q)&=&I(X;W|\tilde U,V,Z),\nonumber
\end{eqnarray}
%I(X;\hat{X}|U,Q) &=& I(X;\hat{X}|\tilde{U})\nonumber \\
%I(Y;U|Q) &=& I(Y;U,Q) = I(Y;\tilde{U}),
%\end{eqnarray}
where $\tilde{U}=(U,Q)$, and in step (a) we used the fact that $Y$
is independent of $Q$. This proves the convexity.
%
%The proof of part~\ref{lemma_properties_R_SI_convex} parallels
%that of part~\ref{lemma_properties_R_convex} of
%Lemma~\ref{lemma:properties_R} and is omitted. For
%part~\ref{lemma_properties_R_SI_size}, using the support
%lemma~\cite{Csiszar81}, the random variable $V$ should have
%$|{\cal Y}|-1$ elements  to preserve $p(y)$, plus three elements
%to preserve $I(V;Y|Z)$, $I(X;U|V,Z)$, and the distortion
%constraint. Once $V$ is fixed, $U$ should have $|{\cal X}||{\cal
%V}|-1$ elements to preserve the joint distribution $p(x,v)$, and
%two more elements to preserve $I(X;U|V,Z)$ and the distortion
%constraint. This completes the proof of the lemma.

To prove Part~\ref{lemma_properties_R_size}, we invoke the support
lemma~\cite[pp. 310]{Csiszar81} three times, each time for one of
the auxiliary random variables $U,V,W$. The external random
variable $U$ must have $|{\cal Y}|-1$ letters to preserve $p(y)$
 plus five more to preserve the expressions
$I(Y;U|Z)$, $I(Z;V|U,X)$, $I(X;W|U,V,Z)$ and the distortions
$\mathbb{E}d_x(X,\hat X(U,V,Z))$ $\mathbb{E}d_z(Z,\hat Z(U,W,X))$.
Note that the joint $p(x,y,z)$ is preserved because of the Markov
form $U-Y-X-Z$, and the structure of the joint distribution given
in (\ref{e_p_twoway}) does not change. We fix $U$, which now has a
bounded cardinality, and we apply the support lemma for bounding
$V$. The external random variable $V$ must have $|{\cal U}||{\cal
Z}|-1$ letters to preserve $p(u,z)$
 plus four more to preserve the expressions
 $I(Z;V|U,X)$, $I(X;W|U,V,Z)$ and the distortions $\mathbb{E}d_x(X,\hat X(U,V,Z))$, $\mathbb{E}d_z(Z,\hat Z(U,W,X))$.
 Note that because of the Markov structure $V-(U,Z)-(X,Y)$ the
 joint distribution $p(u,z,x,y)$ does not change.
 Finally, we fix $U,V$ which now have a bounded cardinality and we apply the support lemma for bounding
$W$. The external random variable $W$ must have $|{\cal U}||{\cal
V}||{\cal X}|-1$ letters to preserve $p(u,v,x)$
 plus two more to preserve the expressions
  $I(X;W|U,V,Z)$ and the distortions  $\mathbb{E}d_z(Z,\hat Z(U,W,X))$.
 Note that because of the Markov structure $W-(U,V,X)-(Z,Y)$ the
 joint distribution $p(u,v,x,y,z)$ does not change.
\end{proof}

\section{Proof of Lemma \ref{l_wyner_ziv_side_gaussian} \label{app_proof_wyner_siz_side}}

% {\it (Gaussian instance of the Wyner-Ziv probwith side information at the encoder and decoder.)}
\label{app_lema_wyner_ziv_side_gaussian} Since $W,X,Z$ are jointly
Gaussian, we have $\mathbb E[X|W,Z]=\alpha W+\beta Z$, for some
scalars $\alpha, \beta$. Furthermore, we have
\begin{equation}
X=\alpha W+\beta Z + N,
\end{equation}
where $N$ is a Gaussian random variable independent of $(W,Z)$ with
zero mean and variance $\sigma^2_{X|W,Z}$. Since $W$ is known to the
encoder and decoder we can subtract $\alpha W$ from $X$, and then
using Wyner-Ziv coding for the Gaussian case
\cite{Wyner78_gaussian_WZ} we obtain
\begin{equation}
R(D)=\frac{1}{2}\log \frac{\sigma^2_{X|W,Z}}{D}.
\end{equation}
Obviously, one can not achieve a rate smaller than this even if $Z$
is known both to the encoder and decoder, and therefore this is the
achievable region.

\bibliographystyle{unsrt}
\bibliographystyle{IEEEtran}
%\bibliography{../GE-channel/ref}
%\bibliography{ref}

\begin{thebibliography}{10}

\bibitem{Kaspi85_two_way}
A.~H. Kaspi.
\newblock Two-way source coding with a fidelity criterion.
\newblock {\em IEEE Trans. Inf. Theory}, 31(6):735--740, 1985.

\bibitem{Wyner_ziv76_side_info_decoder}
A.~D. Wyner and J.~Ziv.
\newblock The rate-distortion function for source coding with side information
  at the decoder.
\newblock {\em IEEE Trans. Inf. Theory}, 22(1):1--10, 1976.

\bibitem{Wyner75_WAK}
A.~D. Wyner.
\newblock On source coding with side-information at the decoder.
\newblock {\em IEEE Trans. Inf. Theory}, 21:294--300, 1975.

\bibitem{Ahlswede-Korner75}
R.~Ahlswede and J.~Korner.
\newblock Source coding with side information and a converse for degraded
  broadcast channels.
\newblock {\em IEEE Trans. Inf. Theory}, 21(6):629--637, 1975.

\bibitem{Kaspi79_dissertation}
A.~Kaspi.
\newblock {\em Rate-distortion for correlated sources with partially separated
  encoders}.
\newblock 1979.
\newblock Ph.{D}. dissertation.

\bibitem{Kaspi_berger82}
A.~Kaspi and T.~Berger.
\newblock Rate-distortion for correlated sources with partially separated
  encoders.
\newblock {\em IEEE Trans. Inf. Theory}, 28:828--840, 1982.

\bibitem{Vasudevan07_helper}
D.\ Vasudevan and E.\ Perron.
\newblock Cooperative source coding with encoder breakdown.
\newblock In {\em Proc. International Symposium on Information Theory (ISIT)},
  Nice, France., June, 2007.

\bibitem{Permuter_steinber_weissman08_helperunpublished}
H.~Permuter, Y.~Steinberg, and T.~Weissman.
\newblock Rate-distortion with a limited-rate helper to the encoder and
  decoder,.
\newblock Availble at http://arxiv.org/abs/0811.4773v1, Nov. 2008.

\bibitem{Berger_Yeung89_one_distortion}
T.~Berger and R.W.~Yeung ~.
\newblock Multiterminal source encoding with one distortion criterion.
\newblock {\em IEEE Trans. Inf. Theory}, 35:228--236, 1989.

\bibitem{Oohama97GaussianMultiterminal}
Y.~Oohama.
\newblock Gaussian multiterminal source coding.
\newblock {\em IEEE Trans. Inf. Theory}, 43:1912--1923, 1997.

\bibitem{Oohama05_Lhelpers_gaussian}
Y.~Oohama.
\newblock Rate-distortion theory for gaussian multiterminal source coding
  systems with several side informations at the decoder.
\newblock {\em IEEE Trans. Inf. Theory}, 51:2577--2593, 2005.

\bibitem{Wagner08_two_sources}
A.~B. Wagner, S.~Tavildar, and P.~Viswanath.
\newblock Rate region of the quadratic gaussian two-encoder source-coding
  problem.
\newblock {\em IEEE Trans. Inf. Theory}, 54:1938--1961, 2008.

\bibitem{Wagner08_L_sources}
S.~Tavildar, P.~Viswanath, and A.~B. Wagner.
\newblock The gaussian many-help-one distributed source coding problem.
\newblock submitted to {\it IEEE Trans. Inf. Theory.} Available at
  http://arxiv.org/abs/0805.1857, 2008.

\bibitem{MaorM_merhav06_two_way}
A.~Maor and N.~Merhav.
\newblock Two-way successively refined joint source-channel coding.
\newblock {\em IEEE Trans. Inf. Theory}, 52(4):1483--1494, 2006.

\bibitem{Pearl00_causality}
J.~Pearl.
\newblock {\em Causality: Models, Reasoning and Inference}.
\newblock Cambridge Univ. Press, 2000.

\bibitem{Kramer03}
G.~Kramer.
\newblock Capacity results for the discrete memoryless network.
\newblock {\em IEEE Trans. Inf. Theory}, IT-49:4--21, 2003.

\bibitem{CovThom06}
T.~M. Cover and J.~A. Thomas.
\newblock {\em Elements of Information Theory}.
\newblock Wiley, New-York, 2nd edition, 2006.

\bibitem{Wyner78_gaussian_WZ}
A.D. Wyner.
\newblock The rate-distortion function for source coding with side information
  at the decoder-{II}: General sources.
\newblock {\em Information and Control}, 38:60--80, 1978.

\bibitem{Steinberg08RateLimitedAtDecoder}
Y.~Steinberg.
\newblock Coding for channels with rate-limited side information at the
  decoder, with applications.
\newblock {\em IEEE Trans. Inf. Theory}, 54:4283--4295, 2008.

\bibitem{Csiszar81}
I.~Csisz{\'a}r and J.~K{\"o}rner.
\newblock {\em Information Theory: Coding Theorems for Discrete Memoryless
  Systems}.
\newblock Academic, New York, 1981.

\end{thebibliography}

\end{document}